\documentclass[11pt]{article}

\usepackage{amsthm}
\usepackage{amsmath}
\usepackage{amsfonts}
\usepackage{amssymb}
\usepackage{bbm}
\usepackage{bm}
\usepackage{color}
\usepackage{hyperref}
\usepackage{mathtools}
\usepackage{algorithm}
\usepackage{algpseudocode} 
\usepackage{setspace}
\usepackage{hyperref}
\usepackage{natbib}
\usepackage{fullpage}

\mathtoolsset{showonlyrefs}

\algrenewcommand\algorithmicrequire{\textbf{Input:}}
\algrenewcommand\algorithmicensure{\textbf{Output:}}

\DeclareFontFamily{U}{mathx}{}
\DeclareFontShape{U}{mathx}{m}{n}{<-> mathx10}{}
\DeclareSymbolFont{mathx}{U}{mathx}{m}{n} 
\DeclareMathAccent{\widecheck}{0}{mathx}{"71}

\newcommand{\ind}{\mathbbm{1}}
\newcommand{\bbE}{\mathbb{E}}
\newcommand{\bbR}{\mathbb{R}}
\newcommand{\bbN}{\mathbb{N}}
\newcommand{\cX}{\mathcal{X}}
\newcommand{\cN}{\mathcal{N}}
\newcommand{\cR}{\mathcal{R}}
\newcommand{\cB}{\mathcal{B}}
\newcommand{\cJ}{\mathcal{J}}
\newcommand{\cI}{\mathcal{I}}
\newcommand{\cY}{\mathcal{Y}}

\newcommand{\sP}{\mathsf{P}}

\newcommand{\sR}{\mathsf{R}}
\newcommand{\Var}{\mathrm{Var}}
\newcommand{\ud}{\mathrm{d}}

\newcommand{\mt}{\mathrm{mt}}
\newcommand{\Gap}{\mathrm{Gap}}
\newcommand{\cont}{\mathrm{ct}} 
\newcommand{\dirac}{\mathrm{dir}}
\newcommand{\GIT}{\mathrm{IT}}
\newcommand{\mh}{\mathrm{mh}}
\newcommand{\tsP}{\tilde{\mathsf{P}}}
\newcommand{\fC}{c}
\newcommand{\bbP}{\mathbb{P}}
\newcommand{\match}{\mathrm{match}}

\newcommand{\cS}{\mathcal{S}}

\newcommand{\ba}{\mathbf{a}}
\newcommand{\bb}{\mathbf{b}}
\newcommand{\bA}{\mathbf{A}}
\newcommand{\bB}{\mathbf{B}}

\newcommand{\tpi}{\tilde{\pi}}
\newcommand{\hpi}{\hat{\pi}}
\newcommand{\halpha}{\hat{\alpha}}
\newcommand{\tP}{\tilde{\sf{P}}}
\newcommand{\bM}{\bm{M}}
\newcommand{\be}{\bm{e}}
\newcommand{\bL}{\bm{L}}
\newcommand{\bY}{\bm{Y}}
\newcommand{\bsbeta}{\bm{\beta}}

\newcommand{\bsgamma}{\bm{\gamma}}
\def\pseu{\mathrm{pm}}

\def\SUPP{Appendices}
\DeclarePairedDelimiter{\norm}{\lVert}{\rVert} 
\def\normc{\mathfrak{C}}

\newtheorem{lemma}{Lemma}
\newtheorem{corollary}{Corollary}
\newtheorem{proposition}{Proposition}
\newtheorem{theorem}{Theorem}
\theoremstyle{definition}
\newtheorem{example}{Example}
\newtheorem{assumption}{Assumption}
\newtheorem{remark}{Remark}
\newtheorem{definition}{Definition}

\title{Importance is Important: Generalized Markov \\ Chain Importance Sampling Methods} 

\author{Guanxun Li$^1$, Aaron Smith$^2$ and Quan Zhou$^{1,}$\footnote{Corresponding author: quan@stat.tamu.edu} \smallskip \\
$^1$ Department of Statistics, Texas A\&M University  \\  
$^2$Department of Mathematics and Statistics, University of Ottawa}
\date{} 

\begin{document}
   
\maketitle

\begin{abstract}
We show that for any multiple-try Metropolis algorithm, one can always accept the proposal and evaluate the importance weight that is needed to correct for the bias without extra computational cost. 
This results in a general, convenient, and rejection-free Markov chain Monte Carlo (MCMC) sampling scheme.  
By further leveraging the importance sampling perspective on Metropolis--Hastings algorithms, 
we propose an alternative MCMC sampler on discrete spaces that is also outside the Metropolis--Hastings framework, 
along with a general theory on its complexity. 
Numerical examples suggest that the proposed algorithms are consistently more efficient than the original Metropolis--Hastings versions.    
\end{abstract}

\noindent%

\section{Introduction}\label{sec:intro}   

\subsection{Markov Chain Importance Sampling} \label{sec:mcis} 
Given a probability density function $\pi$  and a proposal scheme on the same state space,  a standard approach to approximating integrals with respect to $\pi$ is to run a Metropolis--Hastings (M--H) algorithm, which proposes the next state using the given scheme and employs an acceptance-rejection step to ensure that the resulting Markov chain is reversible with respect to $\pi$.  
Most existing Markov chain Monte Carlo (MCMC) algorithms can be formulated under this M--H framework, possibly on augmented spaces.  
The primary purpose of this work is to introduce  general and simple techniques for devising  efficient MCMC algorithms  beyond the M--H framework, which never get stuck at a single state indefinitely.  

At a high level, the proposed algorithms generalize the standard importance sampling method. 
To approximate integrals with respect to $\pi$, we simulate a Markov chain $(x^{(i)})_{i=1}^T$ whose  limiting distribution has density $\tpi$, and for each $x^{(i)}$, we either compute or unbiasedly estimate $\pi (x^{(i)}) / \tpi(x^{(i)})$ up to a normalizing constant. 
Denote this estimated importance weight of $x^{(i)}$ by $w^{(i)}$. 
Then, given any integrable function $f$ of interest, we can estimate $\pi(f) = \int f \, \ud \pi$ using the self-normalized importance sampling estimator  
\begin{equation}\label{eq:iit.estimator} 
 \hat{f}_T  =  \frac{ \sum_{i=1}^T  f(x^{(i)}) w^{(i)} }{ \sum_{i=1}^T  w^{(i)} }. 
\end{equation}     

For all algorithms we propose in this work, $(x^{(i)}, w^{(i)})$ can be thought of as generated from the following process. 
Introduce an auxiliary variable $v$. 
Let $\pi(x, v), \tpi(x, v)$ be two joint distributions that have $\pi(x), \tpi(x)$ as the marginal, respectively. 
We simulate a bivariate Markov chain $(x^{(i)}, v^{(i)})$ with limiting distribution $\tpi(x, v)$ and calculate $w^{(i)}$ by  $ w^{(i)} =  w( x^{(i)}, v^{(i)} )$, where 
\begin{equation}\label{eq:w}
    w(x, v) = C \, \frac{\pi(x, v)}{\tpi(x, v)}, 
\end{equation} 
with $C$ being a fixed but typically unknown constant. 
Hence, $w^{(i)}$ can be viewed as either the exact (un-normalized) importance weight for $(x^{(i)}, v^{(i)})$ 
or a randomized importance weight estimate for $x^{(i)}$; it is unbiased since 
$\int w(x, v) \, \tpi(v \mid x) \ud v = C \pi(x) / \tpi(x)$. 
Under certain regularity conditions on the  Markov chain $(x^{(i)}, v^{(i)})$ and function $f$, the estimator $\hat{f}_T$  converges in probability to $\pi(f)$ and satisfies a central limit theorem; see, e.g.,~\citet[Prop. 5]{deligiannidis2018ergodic}.   
We call such a generalized Markov chain importance sampling scheme ``importance tempering,'' a term coined by~\citet{gramacy2010importance} and also used in~\citet{zanella2019scalable}.

The standard M--H algorithm fits within this framework as well: 
one can let $(x^{(i)})_{i=1}^T$ be the accepted moves and $(w^{(i)})_{i=1}^T$ be the sojourn times. 
So M--H algorithms essentially use acceptance-rejection to estimate the importance weight, 
an interpretation  well known in the MCMC literature~\citep{douc2011vanilla, iliopoulos2013variance, rosenthal2021jump}.  
All the algorithms proposed in this paper essentially replace this acceptance-rejection scheme with a more efficient importance weight estimator. 
 
\subsection{Paper Overview} 
The first algorithm we propose is a   rejection-free modification of the multiple-try Metropolis (MTM) algorithm~\citep{liu2000multiple}; see Section~\ref{sec:mtm}. 
We show that if one simply accepts all proposals and leaves the other steps unchanged, 
the dynamics of the resulting algorithm can be described using the above framework, where the auxiliary variable $v^{(i)}$ is a random set of candidate proposals from which $x^{(i+1)}$ is selected. 
The surprising observation, which makes the proposed algorithm feasible, is that $w^{(i)}$ can be calculated at no additional cost.  
This new algorithm is as flexible as M--H schemes and offers the highly desirable feature of being rejection-free (i.e., $x^{(i+1)} \neq x^{(i)}$ almost surely). 
It generalizes the informed importance tempering (IIT) algorithm introduced in~\citet{zhou2022rapid}, which was itself a generalization of an MCMC sampler for variable selection developed by~\citet{zanella2019scalable} and the rejection-free Metropolis algorithm proposed by~\citet{rosenthal2021jump}. 
We demonstrate the efficiency of our algorithm through simulation studies on both simulated and real data sets. 

Besides the multiple-try proposal,  other MCMC techniques can also be easily integrated with the importance tempering technique. 
We illustrate this in Section~\ref{sec:pseudo.iit} by considering a pseudo-marginal version of IIT, where $\pi(x)$ is replaced by an unbiased estimator of it. 
The auxiliary variable $v$ in this case is the vector of the stationary probability estimates for $x$ and its neighbors.  
We demonstrate the use of our pseudo-marginal algorithm via a numerical example that has been solved using approximate Bayesian computation~\citep{lee2014variance}.  

In Section~\ref{sec:method}, we propose a generalization of M--H schemes on discrete spaces  using a simple technique that effectively combines the advantages of IIT and M--H schemes. 
In most regions of the state space, this algorithm behaves just like the M--H algorithm, but it can escape  local modes by exactly calculating their importance weights. 
Unlike the rejection-free multiple-try algorithm, this algorithm has a simpler dynamics where both $w^{(i)}$ and $x^{(i+1)}$ only depend on $x^{(i)}$. 
Consequently, its behavior is easier to analyze, and we theoretically analyze its complexity in Section~\ref{sec:complexity.two}. A more general theory is developed in Section~\ref{sec:theory} of the \SUPP{}. 

Section~\ref{sec:disc} concludes the paper with further discussion on the comparison between the proposed algorithms and existing ones with similar motivations. Results and more details about simulation studies and all proofs are deferred to the Appendices. 

\subsection{User's Guide} \label{sec:guide}
We describe two typical scenarios for applying the methods studied in this paper. 
First, if one wants to utilize informed proposals which require evaluating $\pi$ of all or some neighboring states (e.g. when one has access to parallel computing), then one can use MT-IT (Algorithm~\ref{alg:mtm-iit}) on  general state spaces and IIT (Algorithm~\ref{alg:iit}) or RN-IIT (Algorithm~\ref{alg:rn-iit}) on discrete spaces. MT-IT is a rejection-free alternative to MTM, and IIT is a rejection-free alternative to the informed M--H algorithm of~\cite{zanella2020informed}. There is almost no downside to using our algorithms, since the rejection-free feature comes at no extra cost. When M--H schemes suffer from low acceptance rates, our samplers can probably offer substantial improvements.  
Indeed, making efficient use of informed proposals is the main motivation behind this work. 

Second, if an existing MCMC sampler defined on a discrete space tends to remain stuck at a single state, then one can use MH-IIT (Algorithm~\ref{alg:mh-iit}) to resolve this issue.  
This is particularly useful when $\pi$ has sharp local modes, in which case conventional MCMC algorithms, including Gibbs samplers, can stay a local mode for an arbitrarily long time.   

For complex sampling tasks, the proposed algorithms can be integrated with other MCMC techniques. As an example, we provide pseudo-marginal IIT (Algorithm~\ref{alg:pm-iit}), which is useful when $\pi(x)$ can only be unbiasedly estimated up to a normalizing constant. 
Other techniques like simulated tempering and non-reversible sampling  can also be used if necessary~\citep{rosenthal2021jump, power2019accelerated}, though they are not explored in this work.

\section{Rejection-free MCMC via Multiple-try Importance Tempering} \label{sec:mtm}

\subsection{Algorithm} \label{sec:mtm-algo}
Let $\pi$ be a probability density function defined on the space $\cX$ with respect to some dominating measure, which we assume can be evaluated up to a normalizing constant. 
For simplicity, we also refer to $\pi$ as a distribution. 
We recall the multiple correlated-try proposal introduced by~\citep{craiu2007acceleration}, which generalizes the multiple-try proposal of~\citet{liu2000multiple}. 
Fix an integer $m \geq 2$, and for each $x \in \cX$, let $q(y_1, y_2,  \dots, y_m \mid x)$ denote the density of a  distribution on $\cX^m$. 
Assume that  $q(\cdot \mid x)$ is exchangeable, but the $m$ coordinates may be correlated.  Denote the marginal distribution of $y_1$ by $q(y_1 \mid x)$, and denote the conditional distribution of $ y_2, \dots, y_{m}$ given $y_1 = y$ by $q(  y_2, \dots, y_{m} \mid x, y_1)$; 
due to the exchangeability, the marginal and conditional distributions of other coordinates are the same. 
We also assume that  $y_1, \dots, y_m$ generated from $q(\cdot \mid x)$ are  distinct and different from $x$ almost surely, and  $q(x \mid y) > 0$ whenever $q(y \mid x) > 0$. 
In many applications, one can simply use an i.i.d. scheme, where $q(y_1, y_2,  \dots, y_m \mid x)$ can be written as $\prod_{i=1}^m q(y_i \mid x)$.  
At state $x$, MTM samples $m$ candidates $y_1, \dots, y_m$ from $q(\cdot \mid x)$ and assigns to each $y_i$ a proposal weight $\alpha(x, y_i) > 0$, which needs to satisfy
\begin{equation}\label{eq:alpha}
   \pi(x) q(y \mid x) \alpha(x, y) =  \pi(y) q(x \mid y) \alpha(y, x). 
\end{equation}  
Let $\cS = \{y_1, \dots, y_m\}$ denote the unordered set of $m$ candidates. 
Given $\cS$, MTM selects $x'$ from $\cS$ with probability proportional to the proposal weight and performs an acceptance-rejection step to decide if the sampler should move to $x'$.  
Exactly calculating the proposal probability from $x$ to $x'$ is difficult, since that would require integrating over $\cS$. 
To circumvent this difficulty, \citet{liu2000multiple} proposed a clever solution that involves generating another set $\cS' = \{x, y'_2, \dots, y'_m\}$ with $y'_2, \dots, y'_m$ sampled from the conditional proposal distribution at $x'$ given that the first coordinate is $x$.  
Explicitly, we can assume $\cS = \{x', y_2, \dots, y_m\}$ without loss of generality, and the density of proposing the pair $(x', \cS')$ given $(x, \cS)$  can be expressed by  
\begin{align}
    p_{\mt}( (x, \cS), (x', \cS') ) =\;&  \frac{ \alpha(x, x')}{ Z(x, \cS)} \, (m - 1)! \, q (y'_2, \dots, y'_m \mid x',   x),  \label{eq:mt.transition} \\
    \text{ where } Z(x, \cS) =\;& \sum\nolimits_{y \in \cS} \alpha(x, y).  \label{eq:def.Zh}
\end{align}   
In~\eqref{eq:mt.transition}, $\alpha(x, x')/ Z(x, \cS)$ is the probability of selecting $x'$ from $\cS$, and $(m - 1)! \, q (y'_2, \dots, y'_m \mid x',  x)$ is the density of proposing the set $\cS'$ at $x'$ conditioned on that $x \in \cS'$. 
The acceptance probability is calculated by comparing $Z(x, \cS)$ and $Z(x', \cS')$. For completeness, we recall MTM in Algorithm~\ref{alg:mtm} in the \SUPP{}. MTM can be further generalized by viewing it as a hit-and-run scheme as described in~\citet{andersen2007hit}.
  
Now let $\sP_{\mt}$ be the bivariate Markov chain with transition density $p_{\mt}$ given in~\eqref{eq:mt.transition}. (If $x' \notin \cS$ or $x \notin \cS'$, then  $p_{\mt}( (x, \cS), (x', \cS') ) = 0$.) 
So $\sP_{\mt}$  can be simulated by following the usual MTM procedure but always accepting the proposed pair $(x', \cS')$.  
Using~\eqref{eq:alpha}, we can find the expression for the stationary distribution of $\sP_{\mt}$ up to a normalizing constant.  

\begin{lemma}\label{lm:mtit} 
The transition kernel $\sP_{\mt}$  is reversible with respect to the joint distribution   
\begin{equation} 
    \pi_{\mt}(x, \cS) \propto \pi(x) Z(x, \cS)   q(y_1, \dots, y_m \mid x). 
\end{equation}  
where $\cS = \{y_1, \dots, y_m\}$ and $Z$ is defined in~\eqref{eq:def.Zh}. 
\end{lemma}

\begin{proof}
See Section~\ref{app:proof} in the \SUPP{}. 
\end{proof}

Imagine that our target is the joint distribution $\pi(x, \cS) = \pi(x)  q(y_1, \dots, y_m \mid x)$, which clearly has $\pi(x)$ as the marginal. 
By Lemma~\ref{lm:mtit}, if we generate samples from $\sP_{\mt}$, the importance weight of $(x, \cS)$ is  given by  $\pi(x, \cS) / \pi_{\mt}(x, \cS) \propto 1 / Z(x, \cS)$. 
Hence, according to the discussion in Section~\ref{sec:mcis},  by simulating $\sP_{\mt}$ and calculating $1 / Z(x, \cS)$, we get a rejection-free MCMC sampler, which we call multiple-try importance tempering (MT-IT) and summarize in Algorithm~\ref{alg:mtm-iit}.  
Compared to MTM, there is no extra computational cost, so the rejection-free feature is obtained free. 
In the simplest case with $m=2$, its cost per iteration is twice that of a vanilla M--H algorithm that uses $q(y \mid x)$ as the proposal. 

There are many choices of $\alpha$ that satisfy~\eqref{eq:alpha}. We assume henceforth that  
\begin{equation}\label{eq:def.local.balance.w}
 \alpha(x, y)=  h\left(  \frac{\pi(y) q(x \mid y) }{\pi(x) q(y \mid x)} \right), 
\end{equation}
for some $h \colon (0, \infty) \rightarrow (0, \infty)$ such that $h(r) = r h(r^{-1})$ for any $r > 0$. We say $h$ is a balancing function, and $\alpha$ is a  locally balanced weighting scheme~\citep{zanella2020informed}.  
\citet{changrapidly} and \citet{gagnon2022improving} showed that this class of weighting functions tends  to yield much faster convergence of MTM than other options such as $\alpha(x, y) = \pi(y) q(x \mid y)$, which has been widely used in earlier implementations of MTM~\citep{yang2021convergence}.    
Examples of balancing functions include $h(r) = 1 + r,  h(r) = 1 \wedge r,  h(r) = 1 \vee r.$

\begin{algorithm}\setstretch{1.2}  
\caption{Multiple-try importance tempering (MT-IT).}\label{alg:mtm-iit}
\begin{algorithmic} 
\Require  Proposal $q$,  weighting function $\alpha$,  neighborhood size $m \geq 2$, number of iterations $T \geq 1$,  initial state $(x^{(0)}, \cS^{(0)})$. 
\For{$i = 0, \dots, T$}
\State Calculate $\alpha(x^{(i)}, y)$  for every $y \in \cS^{(i)}$.
\State Calculate $Z(x^{(i)}, \cS^{(i)}) = \sum_{y \in \cS^{(i)}} \alpha(x^{(i)}, y)$. 
\State Draw $x' \in \cS^{(i)}$ with probability $ \alpha(x^{(i)}, x')/Z(x^{(i)}, \cS^{(i)})$.
\State Draw $(y'_2, \dots, y'_{m})$ from the conditional distribution $q( \cdot \mid x', y_1 = x^{(i)})$.
\State Set $w^{(i)} \leftarrow 1/   Z(x^{(i)}, \cS^{(i)}), \; x^{(i+1)} \leftarrow x', \; \cS^{(i+1)} \leftarrow  \{x^{(i)}, y'_2, \dots, y'_{m}\}$.
\EndFor 
\Ensure  Samples $x^{(1)}, \dots, x^{(T)}$, un-normalized importance weights $w^{(1)}, \dots, w^{(T)}$.
\end{algorithmic}
\end{algorithm}

\subsection{Numerical Example for Multivariate Normal Distributions}\label{sec:mtm-sim}

We first study a multivariate normal example investigated in Section 3 of~\citet{gagnon2022improving}. 
The target distribution $\pi$ is chosen to be the standard $p$-variate normal distribution, $N(0, \bm{I}_p)$, with $p=50$. The proposal distribution $q(\cdot \mid x)$ is given by $N(x, \sigma^2 \bm{I}_p)$, and MCMC samplers are initialized at $x^{(0)} = (10, \dots, 10)$. 
Let $\alpha$ be a locally balanced weighting function as given in~\eqref{eq:def.local.balance.w}. We  consider two balancing functions, the square-root weighting $h(r) = \sqrt{r}$ and Barker's weighting $h(r) = r /(1 + r)$~\citep{barker1965monte, livingstone2022barker}, both recommended by~\citet{gagnon2022improving}. For each choice of $h$, we implement three methods, Algorithm~\ref{alg:mtm-iit} with fixed $\sigma$, MTM with fixed $\sigma$, and an adaptive version of MTM implemented in~\citet{gagnon2022improving}  where $\sigma$ is dynamically tuned.   
For all three methods, we fix the number of tries (i.e., parameter $m$ in Algorithm~\ref{alg:mtm-iit}) equal to $50$, and for the two non-adaptive methods, we fix $\sigma = \sqrt{2.7/p^{0.75}}$, the initial value used in the adaptive method of~\citet{gagnon2022improving}.  
We run all samplers until $5 \times 10^5$ evaluations of $\pi$ have been made, and use the last 50\% of the samples to estimate $M_{2,p} = \int \norm{x}_2^2  \, \pi (\ud x)$, of which the true value is equal to $p$; see Figure~\ref{fig:mtm.normal}. Additional simulation results are provided in Section~\ref{app:mt-it-normal} of \SUPP{}. It was reported in~\citet{gagnon2022improving} that locally balanced MTM schemes have superior performance to   other MTM schemes. But our simulation suggests that  MT-IT performs even better in that it moves to high-density regions faster and yields more accurate estimates.    

\begin{figure}
    \centering
    \includegraphics[width=0.75\linewidth]{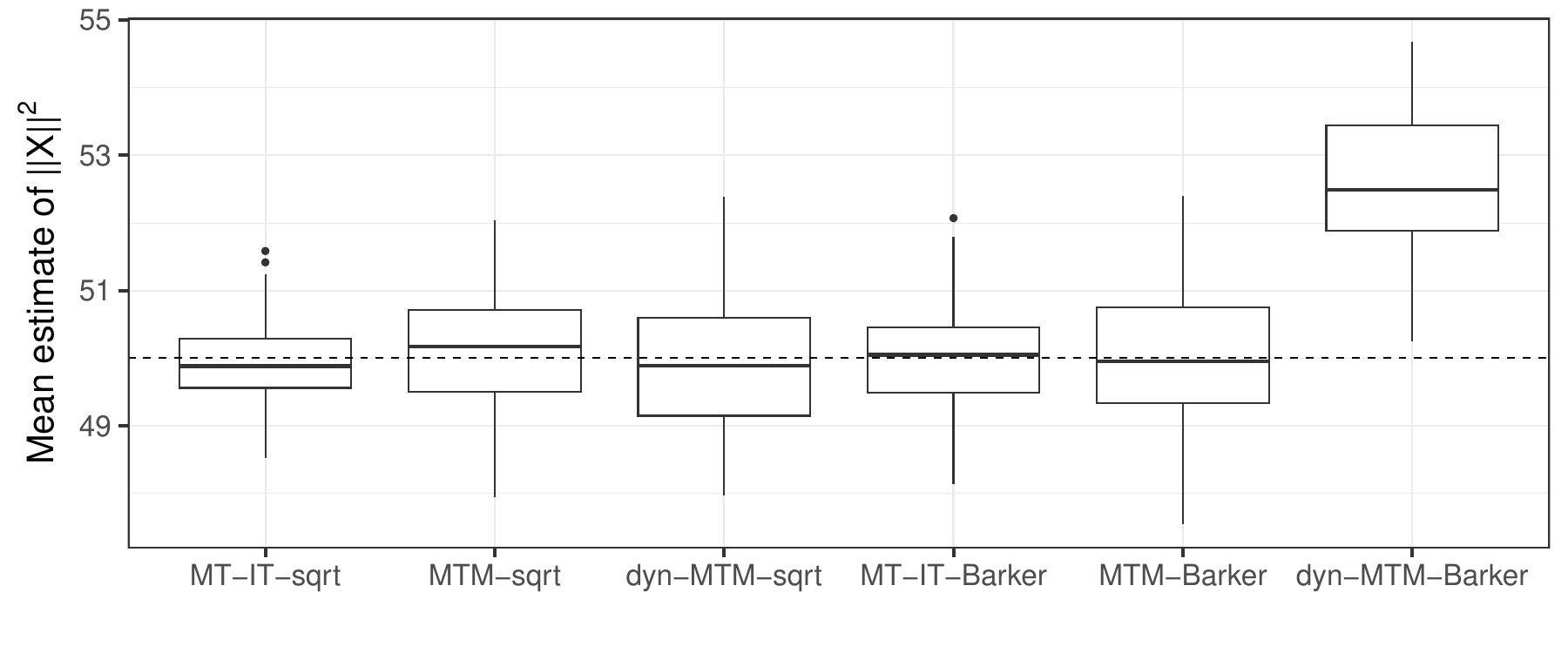}
    \caption{Box plot of the estimate of $M_{2, p}$ over 100 runs using  half of the samples. }
    \label{fig:mtm.normal}
\end{figure}
 
\subsection{Application of MT-IT to Record Linkage}\label{sec:record}

We consider the application of the MT-IT algorithm to a bipartite record linkage problem studied in~\citet{zanella2020informed}.  The data is available in the \texttt{R} package \texttt{italy}.  
The goal is to generate samples from the  joint posterior distribution  $\pi(\bM, p_{\mathrm{match}}, \lambda)$ considered in~\citet{zanella2020informed},  where $p_{\mathrm{match}}, \lambda$ are two continuous hyperparameters and $\bM$ is a discrete parameter of interest. Specifically, $\bM = (M_1, \dots, M_{n_1})$ takes values in  $\{0, 1, \dots, n_2\}^{n_1}$, where $n_1, n_2$ denote the sample sizes of two data sets. If $M_i = j > 0$, it means the $i$-th record in the first data set is matched with the $j$-th record in the second. 
A standard approach  is to use the Metropolis-within-Gibbs sampler that alternates between direct sampling from $\pi( p_{\mathrm{match}}, \lambda \mid \bM)$  and  an M--H scheme targeting $\pi( \bM \mid p_{\mathrm{match}}, \lambda )$.  
We devise an MT-IT scheme using a mixture proposal distribution: given current state $(\bM, p_{\mathrm{match}}, \lambda)$, with probability $\delta$ we propose a new value for $(p_{\mathrm{match}}, \lambda)$ from its conditional distribution given $\bM$, and with probability $1 - \delta$ we propose a new value for $\bM$ by uniformly sampling its neighborhood. 
More details of the samplers and numerical results are given in Section~\ref{sec:real} of \SUPP{}. 
We find that, regardless of the sample size,  MT-IT always moves to high-posterior regions more quickly than M--H schemes.   

\subsection{MT-IT on Discrete Spaces} \label{sec:mtm-iit}
Algorithms such as MTM and MT-IT are especially useful on discrete spaces, since gradient-based samplers (e.g. Metropolis-adjusted Langevin algorithm) are not applicable, and typically one needs to evaluate the stationary probabilities of states in the current neighborhood to figure out in which direction $\pi$ is likely to increase. Below we describe two special cases  of Algorithm~\ref{alg:mtm-iit} on discrete spaces. We use $\cN_x$ to denote the set of neighboring states around $x$. 
Assume that  $y \in \cN_x$ whenever $x \in \cN_y$, and $x \notin \cN_x$,  $|\cN_x| \geq 2$ for each $x$. 
 
\begin{example}\label{ex:IIT} 
Assume that $|\cN_x| = N$ for all $x$, and for each $x$, let $q(y_1, \dots, y_m \mid x)$  simply be a random permutation of states in $\cN_x$. 
Then, $q(y \mid x) = q(x \mid y) = 1 / N$ for every $y \in \cN_x$, and $Z(x, \cS)$  only depends on $x$. 
In this case, MT-IT is reduced to the IIT algorithm described in~\citet{zhou2022rapid}, which we recall in Algorithm~\ref{alg:iit}. 
The un-normalized importance weight for state $x$ is $1/Z(x)$ where $Z(x) = \sum_{y \in \cN_x} \alpha(x, y)$. It is straightforward to show that Algorithm~\ref{alg:iit} is still correct if we drop the assumption that $|\cN_x|$ is a constant. 
\end{example}  

\begin{algorithm}\setstretch{1.2}  
\caption{Informed importance tempering (IIT).}\label{alg:iit}
\begin{algorithmic} 
\Require Balancing function $h$, number of iterations $T \geq 1$, initial state $x^{(0)} \in \cX$. 
\For{$i = 0, \dots, T$}
\State Calculate $\alpha(x^{(i)}, y) = h( \pi(y)/\pi(x^{(i)}) )$  for every $y \in \cN(x^{(i)})$.
\State Calculate $Z(x^{(i)}) = \sum_{y \in \cN(x^{(i)})} \alpha(x^{(i)}, y)$.
\State Draw $x' \in \cN(x^{(i)})$ with probability $ \alpha(x^{(i)}, x')/Z(x^{(i)})$.
\State Set $w^{(i)} \leftarrow 1/Z(x^{(i)}), \; x^{(i+1)} \leftarrow x'$.
\EndFor 
\Ensure  Samples $x^{(1)}, \dots, x^{(T)}$, un-normalized importance weights $w^{(1)}, \dots, w^{(T)}$.
\end{algorithmic}
\end{algorithm}

To our knowledge, the first IIT algorithm was the tempered Gibbs sampler proposed by~\citet{zanella2019scalable}, which can be seen as the application of IIT to Bayesian variable selection with $\alpha(x, y) = 1 + \pi(y)/\pi(x)$. 
\citet{rosenthal2021jump} proposed the rejection-free Metropolis algorithm, which can be seen as IIT with $\alpha(x, y) = \min\{1, \pi(y) / \pi(x)\}$ (see also Remark~\ref{rmk:other.mh.it}). 
Similar ideas can be traced back to kinetic Monte Carlo algorithms~\citep{bortz1975new, dall2001faster}, which are also known as  ``$n$-fold way'' and ``waiting time method.'' These methods simulate a continuous-time Markov chain whose stationary distribution coincides with $\pi$, and they can be recast as variants of IIT where the importance weight for $x$ is generated as an exponential random variable with mean $1/Z(x)$. 
More efficient, non-reversible versions have been proposed in~\citet{hamze2007large} and~\citet{power2019accelerated}. 
 
\begin{example}\label{ex:RN-IIT}
Let $m \leq \min_x |\cN_x|$. 
For each $x$, let $q(y_1, \dots, y_m \mid x)$ be  the distribution of a simple random sample of size $m$ drawn from $\cN_x$ without replacement, 
Hence, $q(y \mid x) =   m / |\cN_x|$ for every $y \in \cN_x$, and 
$q(y_2, \dots, y_m \mid x, y)$ is the distribution of a simple random sample of size $m - 1$ drawn from $\cN_x \setminus \{y\}$ without replacement. For a locally balanced weighting scheme, we can express the weight by 
\begin{equation}
    \alpha(x, y) = h \left( \frac{\pi(y) |\cN_x|}{\pi(x) |\cN_y|} \right).
\end{equation}
We refer to MT-IT in this case as the random neighborhood informed importance tempering (RN-IIT) sampler, which we summarize in Algorithm~\ref{alg:rn-iit} below. 

RN-IIT has the advantage of not requiring evaluation of $\pi$ over the entire neighborhood $\cN_x$, which can be time consuming if $\cN_x$ is large. 
In Section~\ref{sec:sim}, we consider a variable selection problem where the posterior distribution $\pi$ is highly multimodal, and RN-IIT achieves the best performance among all methods we have tried.   
\end{example}  

\begin{algorithm}\setstretch{1.2}  
\caption{Random neighborhood informed importance tempering (RN-IIT)}\label{alg:rn-iit}
\begin{algorithmic} 
\Require Balancing function $h$, neighborhood size $m\geq 2$, number of iterations $T\geq 1$, initial state $(x^{(0)}, \cS^{(0)}) $.
\For{$i = 0, \dots, T$}
\State Calculate $\alpha(x^{(i)}, y) = h \left( \frac{\pi(y) |\cN(x^{(i)})|}{\pi(x^{(i)}) |\cN_y|} \right)$ for every $y \in \cS^{(i)}$.
\State Calculate $Z(x^{(i)}, \cS^{(i)}) = \sum_{y\in\cS^{(i)}}\alpha(x^{(i)}, y)$.
\State Draw $x'\in \cS^{(i)}$ with probability $\alpha(x^{(i)}, x') / Z(x^{(i)}, \cS^{(i)})$.
\State Draw a simple random sample $(y_2', \dots, y_{m}')$  without replacement from $\cN_{x'} \setminus \{x\} $.
\State Set $w^{(i)} \leftarrow 1/  Z(x^{(i)}, \cS^{(i)})$, $x^{(i+1)}\leftarrow x'$, $\cS^{(i+1)}\leftarrow\{x^{(i)}, y_2', \cdots, y_{m}'\}$.
\EndFor 
\Ensure  Samples $x^{(1)}, \dots, x^{(T)}$, un-normalized importance weights $w^{(1)}, \dots, w^{(T)}$.
\end{algorithmic}
\end{algorithm}

\begin{remark}\label{rmk:gibbs}
RN-IIT is significantly different from Gibbs sampling, though both algorithms generate the next state $x'$ by drawing it from a set $\cS$ around the current state $x$. 
For RN-IIT, $\cS$ is a random subset of $\cN_x$ (and recall $x \notin \cN_x$).  
For Gibbs sampling, we require $\cX = \prod_k \cX_k$ is a product space and have $\cS = \{x\} \cup \cN_k(x)$ for some $k$, where $\cN_k(x)$ is the collection of all states that differ from $x$ only at the $k$-th coordinate.  
The key distinction is that a Gibbs sampler can still get stuck at a single state (e.g., the global mode of $\pi$),   
while RN-IIT always leaves the current state and assigns a large importance weight to a local mode.    
An algorithm similar to RN-IIT, but with $\cS = \cN_k(x)$, has been cleverly crafted in~\cite{chen2023sampling}. It is more complicated than RN-IIT due to the difficulty in ensuring the detailed balance condition. 
\end{remark}

\begin{remark}\label{rmk:m}
In this work we do not address the optimal selection of the parameter $m$ for MT-IT and RN-IIT.
This is a very difficult problem of great interest from both theoretical and practical perspectives.  
Similar questions have been investigated for pseudo-marginal MCMC methods~\citep{bornn2017use, sherlock2017pseudo}, and the theoretical tools developed in those works may be useful for understanding the effect of $m$ on the efficiency of MT-IT and RN-IIT.  
\end{remark}

\section{Pseudo-marginal Informed Importance Tempering} \label{sec:pseudo.iit} 
\subsection{Algorithm}\label{sec:pseudo-alg} 
Algorithms~\ref{alg:mtm-iit} and~\ref{alg:iit} can be further extended by utilizing other MCMC techniques. 
For example, for  doubly intractable target distributions, we can integrate importance tempering with pseudo-marginal methods~\citep{andrieu2009pseudo}.  Here,  ``doubly intractable'' means that the evaluation of $\pi$  up to a normalizing constant involves another intractable integral~\citep{murray2006mcmc}.  
In this section, we consider a pseudo-marginal extension of IIT on discrete spaces, which has also been numerically studied in~\cite{rosenthal2021jump}. 
The pseudo-marginal version of MT-IT on general state spaces can be devised similarly. 

We follow~\citet{andrieu2015convergence} to assume that, for each $x$, we have an unbiased estimator $\hat{\pi}(x)$, and we can express it by $\hat{\pi}(x) = U_x \pi(x) $, where $U_x > 0$ is unknown with distribution $g(\cdot \mid x)$ such that $ \int u \, g(u \mid x) \ud u = 1$. 
Then we calculate the proposal weight of each $x$  as in Algorithm~\ref{alg:iit}, but with $\pi(x)$  replaced by $\hat{\pi}(x)$. That is, we pick a balancing function $h$ and set 
\begin{equation}\label{eq:alpha.pseudo}
   \halpha(x, y ) =  h \left( \frac{\hpi(y)}{\hpi(x)} \right).
\end{equation} 
The resulting pseudo-marginal version of IIT is described in Algorithm~\ref{alg:pm-iit}.   
Note that when we move from $x$ to $x'$,  both $\hat{\pi}(x)$ and $\hat{\pi}(x')$ are saved and will be re-used in the next iteration. This step is crucial for the correctness of the algorithm.   
All other estimates, $\hat{\pi}(y)$ for $y \in \cN_x \setminus \{x'\}$,  can be thrown out as the algorithm runs---they are always resampled before being used next time.  
 
To show that Algorithm~\ref{alg:pm-iit} is correct, we  again use the Markov chain importance sampling framework described in Section~\ref{sec:mcis}. 
The auxiliary variable is $(\hat{\pi}(x), \hat{\pi}(y_1), \dots, \hat{\pi}(y_{N}))$, where $N = |\cN_x|$ and $\{y_1, \dots, y_{N}\} = \cN_x$. Since we assume $\hat{\pi}(x) = U_x \pi(x)$, equivalently, we can denote this auxiliary variable by $ (u_x, u_1, \dots, u_{N})$. 
Consider the transition from 
$(x,   u_x, u_1, \dots, u_{N})$ to $(x',  u_{x'}, u'_1, \dots, u'_{N'}) $, where $N' = |\cN_{x'}|$,  and 
denote the neighbors of $x'$ by $\{y'_1, \dots, y'_{N'}\} = \cN_{x'}$.
Since the labeling of the states in $\cN_x$ and $\cN_{x'}$ is irrelevant, 
 we can assume $x' = y_1, x = y'_1, u_x = u'_1, u_1 = u_{x'}$ without loss of generality. Then, 
\begin{equation}\label{eq:pm-transition}
   p_{\pseu}(  (x,   u_x, u_1, \dots, u_N), (x',  u_{x'}, u'_1, \dots, u'_{N'})  ) 
= \frac{ \halpha(x, y_1; u_x, u_1) }{ Z(x, u_x, u_1, \dots, u_N)}  \prod_{i = 2}^{N'} g(u'_i \mid y'_i), 
\end{equation}
where the notation $\halpha(x, y_1; u_x, u_1)$ indicates that $\halpha$ is evaluated by using the estimates $\hpi(x) = u_x \pi(x)$ and $\hpi(y_1) = u_{1} \pi(y_1)$, and   
\begin{equation}\label{eq:def-Z-pm}
    Z(x, u_x, u_1, \dots, u_N)= \sum_{i = 1}^{N}  \halpha(x, y_i; u_x, u_i).  
\end{equation} 

\begin{lemma}\label{lm:pm}
The transition kernel defined by~\eqref{eq:pm-transition} is reversible with respect to the  distribution 
\begin{equation}\label{eq:def-station-pm}
     \pi_{\pseu}(x, u_x, u_1, \dots, u_N) \propto  u_x \, \pi(x) \, g(u_x \mid x) \, Z(x, u_x, u_1, \dots, u_N) \, \prod_{i=1}^{N} g( u_i \mid y_i), 
\end{equation}
where $\{y_1, \dots, y_N\} = \cN_x$. 
\end{lemma}
\begin{proof}
See Section~\ref{app:proof} in the \SUPP{}.
\end{proof}

By Lemma~\eqref{lm:pm}, $1 / Z(x, u_x, u_1, \dots, u_N)$ is the un-normalized importance weight for target joint distribution $\pi(x, u_x, u_1, \dots, u_N) = u_x   \pi(x)   g(u_x \mid x)  \prod_{i=1}^{N} g( u_i \mid y_i),$  which has $\pi(x)$ as the marginal. This establishes the correctness of Algorithm~\ref{alg:pm-iit}.  

\begin{algorithm}\setstretch{1.2}  
\caption{Pseudo-marginal informed importance tempering (P-IIT).}\label{alg:pm-iit}
\begin{algorithmic} 
\Require Balancing function $h$, unbiased estimator $\hat{\pi}$, number of iterations $T \geq 1$,  initial state $x^{(0)} \in \cX$, $x^{(1)} \in \cN(x^{(0)})$.
\State Generate $\hat{\pi}(x^{(0)})$ and $\hat{\pi}(x^{(1)})$.
\For{$i = 0, \dots, T$}
\State Calculate $\hat{\alpha}( x^{(i)}, x^{(i-1)}) = h( \hpi(x^{(i-1)})/ \hpi(x^{i}))$.
\For{$y \in \cN(x^{(i)}) \setminus \{x^{(i-1)}\}$}
\State Generate $\hat{\pi}(y)$ and calculate $\hat{\alpha}( x^{(i)}, y) = h( \hpi(y)/ \hpi(x^{i}))$.
\EndFor
\State Calculate $\hat{Z}(x^{(i)}) = \sum_{y \in \cN(x^{(i)})} \hat{\alpha}(x^{(i)}, y)$.
\State Draw $x' \in \cN(x^{(i)})$ with probability $ \hat{\alpha}(x^{(i)}, x')/\hat{Z}(x^{(i)})$.
\State Set $w^{(i)} \leftarrow 1/\hat{Z}(x^{(i)}), \; x^{(i+1)} \leftarrow x'$. 
\State Save $\hpi(x^{i}), \hpi(x^{(i+1)})$.
\EndFor 
\Ensure  Samples $x^{(1)}, \dots, x^{(T)}$, un-normalized importance weights $w^{(1)}, \dots, w^{(T)}$.
\end{algorithmic}
\end{algorithm}

\subsection{An Example in Approximate Bayesian Computation}\label{sec:abc}
In approximate Bayesian computation, the target distribution $\pi$  is an approximation to the  true  posterior distribution and is constructed by repeatedly simulating the data and comparing it to the observed data~\citep{marin2012approximate}. In general, $\pi$ is doubly intractable and pseudo-marginal M--H schemes are often invoked to evaluate integrals with respect to $\pi$. A known issue of this approach is that, especially in the tails of $\pi$, the sampler can get stuck at some $x$ when we happen to generate some $\hat{\pi}(x) \gg \pi(x)$. 
In Section~\ref{app:pse} of the \SUPP{}, we study a toy example presented in Section 4.2 of~\citet{lee2014variance}, where $\pi$ is a geometric distribution with success probability $1 - ab$ for some $a, b\in (0, 1)$. 
For each $x \in \{1, 2, \dots, \}$,   $\hat{\pi}(x)$ is generated by $\hat{\pi}(x) = \tilde{U}_x (1-a)a^{x - 1} $, where $\tilde{U}_x$ is the mean of $K$ Bernoulli random variables with success probability $b^x$. 
Our simulation study   shows that Algorithm~\ref{alg:pm-iit} outperforms the pseudo-marginal M--H algorithm by a very wide margin, when both samplers are initialized at very large values.

\section{Fast Importance Estimation on Discrete Spaces}\label{sec:method} 
\subsection{Algorithm} \label{sec:mh-iit-algo}
In this section, we introduce another MCMC sampler  which generalizes IIT (i.e., Algorithm~\ref{alg:iit}). 
To motivate it, we first make an observation that connects  IIT to  M--H algorithms.  
Assume that $\alpha(x, y) = h(\pi(y) / \pi(x))$ takes values in $[0, 1]$. 
The un-normalized  importance weight of each $x$ in IIT is given by $1 / Z (x)$, where $Z(x) = \sum_{y \in \cN_x} \alpha (x, y)$. 
In IIT, we always exactly calculate this weight.  
Alternatively, we can apply the acceptance-rejection method where
 $y$ is repeatedly sampled from $\cN_x$  uniformly and accepted with probability $\alpha(x, y)$.  
Denote by $\tau_{\rm{ar}}(x)$ the number of acceptance-rejection iterations needed until a  proposal is accepted. 
It is clear that $\tau_{\rm{ar}}(x)$ is a geometric random variable with success probability 
$Z(x)/ |\cN_x|$, and that the accepted proposal has the same distribution as the next state $x'$ generated in Algorithm~\ref{alg:iit}.  
Such an acceptance-rejection scheme is essentially the standard M--H algorithm, which only differs from IIT in that the weight $1 / Z (x)$ is unbiasedly estimated by $\tau_{\rm{ar}}(x)/ |\cN_x|$.  

\begin{remark}\label{rmk:mh.hastings} 
In the original work of~\citet{hastings1970monte}, the acceptance probability function $\alpha$ is only required to take values in $[0, 1]$ and satisfy~\eqref{eq:alpha}.   
The choice  $\alpha(x, y) = \min\{1, \pi(y) / \pi(x)\}$ (assuming $q(x,y) = q(y,x)$ for any $x\neq y$) became prevalent because it was  shown to be optimal for M--H schemes~\citep{peskun1973optimum}, a result known as Peskun's ordering.  
By~\eqref{eq:def.local.balance.w}, this corresponds to using balancing function $h(r) = 1 \wedge r$. 
However, for IIT, this choice is often too conservative in the sense that all neighboring states $y$ with $\pi(y) \geq \pi(x)$ are treated equally. 
The optimal choice for IIT is very difficult to characterize. \citet{zhou2022rapid} recommends using balancing function $h(r) = \sqrt{r}$, which is more aggressive than  $h(r) = 1 \wedge r$ but offers much more robust performance than other aggressive choices such as $h(r) = 1 + r$.  
\end{remark} 

We now have two methods for estimating the importance weight $1/Z(x)$ at $x$: the exact calculation used by  IIT and the acceptance-rejection method used by M--H schemes. The former always requires $|\cN_x|$ evaluations of $\pi$, while the latter, on average, requires $|\cN_x| / Z(x)$ evaluations. Consequently, when $Z(x) < 1$, the exact method has a smaller computational cost.  
For most states on $\cX$, we should have $Z(x) > 1$, but a local mode $\tilde{x}$ may have $Z(\tilde{x}) \ll 1$, which can trap the M--H algorithm for a long time.  
This motivates us to propose a new estimator by combining the two methods. 
In each iteration of the acceptance-rejection method, with probability $\rho$ (which may depend on $x$), we terminate the iteration and exactly calculate $|\cN_x|/Z(x)$, which is the expected number of additional acceptance-rejection iterations needed to leave $x$ (due to the memoryless property of $\tau_{\rm{ar}}$).   
This scheme is described in Algorithm~\ref{alg:mix.weight}, 
and the resulting importance tempering sampler is given in Algorithm~\ref{alg:mh-iit}, which we call MH-IIT.  
In the next subsection, we present simulation studies which show that using a small positive value of $\rho$ often works well, and MH-IIT, as an interpolation of M--H and IIT schemes, can significantly outperform both M--H and IIT in some scenarios. 

\begin{algorithm}
\setstretch{1.2}  
\caption{Importance Weight Estimation}\label{alg:mix.weight}
\begin{algorithmic} 
\Require Balancing function $h \in [0, 1]$, IIT update frequency $\rho$, current state $x$.
\State Set $y \gets x$ and $w \gets 0$.
\While{$y = x$}
    \State Draw $U \sim \mathrm{Unif}(0, 1)$.
    \If{$U \leq \rho$}
        \State Calculate $\alpha(x, y) \gets h\left(\frac{\pi(y)}{\pi(x)}\right)$ for every $y \in \cN_x$.
        \State Calculate $Z(x) \gets \sum_{y \in \cN_x} \alpha(x, y)$.
        \State Set $w \gets w + \frac{|\cN_x|}{Z(x)}$.
        \State Draw $y$ with probability $\frac{\alpha(x, y)}{Z(x)}$.
    \Else
        \State Set $w \gets w + 1$.
        \State Draw $y$ uniformly from $\cN_x$ and calculate $\alpha(x, y) \gets h\left(\frac{\pi(y)}{\pi(x)}\right)$.
        \State Set $y \gets x$ with probability $1 - \alpha(x, y)$.
    \EndIf
\EndWhile
\State Set $w \gets \frac{w}{|\cN_x|}$. 
\Ensure Importance weight estimate $w$ for state $x$, and next state $y$.
\end{algorithmic}
\end{algorithm}

\begin{algorithm}\setstretch{1.2}  
\caption{Metropolis--Hastings-boosted informed importance tempering (MH-IIT).}\label{alg:mh-iit}
\begin{algorithmic} 
\Require Balancing function $h \in [0, 1]$,  function $\rho \colon \cX \rightarrow [0, 1]$, number of iterations $T \geq 1$, initial state $x^{(0)} \in \cX$.
\For{$i = 0, \dots, T$}
\State Set $(w^{(i)}, x^{(i+1)}) \leftarrow$ Algorithm~\ref{alg:mix.weight} $(  h, \rho(x), x^{(i)})$.
\EndFor 
\Ensure  Samples $x^{(1)}, \dots, x^{(T)}$, un-normalized importance weights $w^{(1)}, \dots, w^{(T)}$.
\end{algorithmic}
\end{algorithm}

To show that Algorithm~\ref{alg:mh-iit} is an importance tempering scheme targeting $\pi$, it suffices to verify that the samples of Algorithm~\ref{alg:mh-iit} form a Markov chain with stationary distribution $\tpi$ such that $\pi(x)/\tpi(x)$ is unbiasedly estimated up to a normalizing constant by Algorithm~\ref{alg:mix.weight}. 
These are proved in Lemma~\ref{lemma:mh-iit} and  Lemma~\ref{lemma:exp-weight} below respectively.  
For later use, in Lemma~\ref{lemma:exp-weight} we also derive the expected computational cost of Algorithm~\ref{alg:mix.weight} in terms of the number of evaluations of $\pi$. 
 
\begin{lemma}\label{lemma:mh-iit}
    The samples $x^{(1)}, \dots, x^{(T)}$ generated from Algorithm~\ref{alg:mh-iit} form a Markov chain with stationary distribution $\tpi(x) \propto \pi(x) Z(x)$, where $Z(x) = \sum_{y \in \cN_x} h(\pi(y) / \pi(x) )$. 
\end{lemma}
\begin{proof}
See Section~\ref{app:proof} in the \SUPP{}.
\end{proof}

\begin{lemma}\label{lemma:exp-weight}
Let $W = W(x)$ denote the importance weight estimate  generated by Algorithm~\ref{alg:mix.weight} with input $(h, \rho, x)$. 
Let $Z = Z(x)$ and $N = |\cN_x|$. 
Then, $\bbE[W ] = 1 / Z  $, and 
\begin{equation}\label{eq:var.W}
    \Var(W ) =  \frac{ (1 - Z/N)(1 - \rho )}{Z^2 + \rho Z (N - Z)}. 
\end{equation}
Let $K = K(x)$ denote the number of evaluations of $\pi$ involved in Algorithm~\ref{alg:mix.weight}, assuming that the  ``if'' part requires $N$ evaluations and the ``else'' part requires one. Then, 
\begin{equation}\label{eq:cost.mix}
\bbE[K] = \frac{ \rho  (  N - 1  ) + 1 }{  \rho  (1 -  Z  / N)  + Z / N}.
\end{equation} 
\end{lemma}
\begin{proof}
See Section~\ref{app:proof} in the \SUPP{}.
\end{proof}

\begin{remark}\label{rmk:compare.var.mh.iit}
     $\Var(W )$ is monotone decreasing with $\rho$. When $\rho = 0$ (e.g., in an M--H scheme), $\Var(N W ) =  N(N - Z)/Z^2$, which is the variance of a geometric random variable with mean $N/Z$. When $\rho = 1$, $\Var(W) = 0$ since the importance weight is  calculated exactly. 
\end{remark}

\begin{remark}\label{rmk:other.mh.it}
The importance sampling perspective on M--H algorithms has  been well studied in the literature. The exact calculation of importance weights on discrete spaces was proposed by~\cite{rosenthal2021jump}, and as they noted, this can be implemented very efficiently when parallel computing is available. IIT is a simple extension of their algorithm that allows for tuning of the balancing function $h$. 
Other methods for reducing the importance weight estimator  have also been proposed~\citep{casella1996rao, sahu2003self, atchade2005improving,  jacob2011using, douc2011vanilla, iliopoulos2013variance}. 
But these works consider general state spaces where exact calculation of  importance weights is not an option, and their methods either have limited application (e.g. can only be used for independence M--H algorithms) or are computationally very expensive. 
Moreover, all these works seem to have only considered the special case $h(r) = \min\{1, r\}$ (i.e., the acceptance probability function used in the standard M--H schemes).  
\end{remark}

\subsection{Simulation Studies on Three Toy Models} \label{sec:sim.toy}
In Section~\ref{app:toy-varsel} of the \SUPP{}, we present simulation studies for three toy examples on  $\cX = \{0, 1\}^p$.  We assume $\pi$ takes the following forms: 
\begin{align}
    \pi_{\mathrm{uni}}(x) \propto\;&  \exp(-\theta \norm{x - x^*}_1),    \\ 
    \pi_{\mathrm{dep}}(x) \propto\;& \exp\left[- \theta \left\{ (\norm{x}_1 - 1) \ind_{\{x_1 = 1\}} + (2p - \norm{x}_1) \ind_{\{x_1 = 0\}} \right\}\right], \\ 
    \pi_{\mathrm{bi}}(x) \propto\;&  \exp(-\theta \norm{x - x_{(1)}^*}_1) + \exp(-\theta \norm{x - x_{(2)}^*}_1), 
\end{align}
where $\norm{\cdot}_1$ denotes the $\ell^1$-norm, $\theta > 0$ is a tuning parameter, and $x^*, x_{(1)}^*, x_{(2)}^* \in \cX$ denote  fixed locations of the modes. 
These choices of $\pi$ correspond to three typical scenarios in variable selection:  a unimodal posterior with independent coordinates,  a unimodal posterior with dependent coordinates, and a bimodal posterior.  
The key advantage of these choices is that  the normalizing constant can be computed in closed form, allowing us to numerically evaluate the total variation distance between the target distribution and the importance-weighted empirical distribution of MCMC samples. 
The efficiency of each sampler is measured by the number of evaluations of $\pi$ needed to approximately reach stationarity; see Section~\ref{app:toy-varsel} \SUPP{} for precise definitions. 
We summarize our main findings below. 

\begin{enumerate}
    \item In all three settings,   when $\theta$ is small, the target distribution is flat and the standard M--H algorithm works well. When $\theta$ is large, IIT, RN-IIT and MH-IIT show significant advantages over M--H and other competing algorithms, but the best performer depends on the setting. 
    \item For $\pi_{\mathrm{uni}}$ with large $\theta$, MH-IIT with $\rho(x) = 0.025$ and balancing function $h(r) = 1 \wedge r$ is overwhelmingly more efficient than all other algorithms including M--H and IIT. 
    \item For $\pi_{\mathrm{dep}}$, another MH-IIT scheme with a more aggressive balancing function performs much better than the MH-IIT scheme with $h(r) = 1 \wedge r$.  This shows that the optimal choice of $h$ for MH-IIT schemes depends on the problem, which is consistent with our discussion on Peskun's ordering in Remark~\ref{rmk:mh.hastings}.  
    A more detailed analysis on the convergence rate of MH-IIT schemes for $\pi_{\mathrm{dep}}$ will be given in Section~\ref{sec:example.complex}.  
    \item The performance of RN-IIT is least sensitive to the choice of $\theta$ and the setting. Further, it is always more efficient than the locally balanced MTM algorithm~\citep{changrapidly}, which has a very similar dynamics but does not use importance weighting. 
\end{enumerate}


\subsection{Simulation Study on Variable Selection with Dependent Designs} \label{sec:sim}

Our last simulation study considers a more realistic  setting for Bayesian variable selection, which has often been studied in the literature~\citep{yang2016computational, zhou2022dimension, changrapidly}. 
First, we generate the design matrix $\bL \in \bbR^{n \times p}$ with $n = 1,000$ and $p = 5,000$ by sampling each row independently from the multivariate normal distribution $N(0, \bm{\Sigma})$, where $\Sigma_{ij} = e^{-|i - j|}$. The response vector is generated by $\bY =  \bL \bsbeta  + \be$, where $\be \sim N(0, \bm{I})$, $\beta_j = 0$ if $j > 20$,  and for $j = 1, \dots, 20$, $\beta_j$ is drawn independently from $2 \sqrt{ (\log p )/n } \, \mathrm {Unif}\bigl((2, 3)\cup (-3, -2)\bigr)$. This is called the intermediate SNR (signal-to-noise ratio) case in~\citet{yang2016computational}.  
The parameter of interest is $\bsgamma \in \{0, 1\}^p$, indicating which predictor variables are selected (i.e., have nonzero regression coefficients), and we calculate its posterior distribution $\pi$ in the same way as in~\citep{yang2016computational, zhou2022dimension}.  It has been shown in those works that $\pi$ tends to be highly multimodal in the intermediate SNR case, making the MCMC sampling challenging. 

We consider nine MCMC samplers, details of which are given in Section~\ref{supp:var-sel-real} of the \SUPP{}. Each sampler is initialized at the same starting point, $\bsgamma^{(0)}$, with 10 randomly selected covariates, and run each sampler until $2.5 \times 10^6$ evaluations of $\alpha$ have been made. Denote by $\hat{\bsgamma}$ the model with the highest posterior probability among all models visited by any of the samplers, and record the number of evaluations of $\alpha$ taken by each sampler to find $\hat{\bsgamma}$. We repeat this simulation 100 times, and show the results in Figure~\ref{fig:vs.sim}. The advantage of IIT methods over the others is quite significant, and RN-IIT has the best performance.  
  
\begin{figure}
    \centering
    \includegraphics[width=0.8\linewidth]{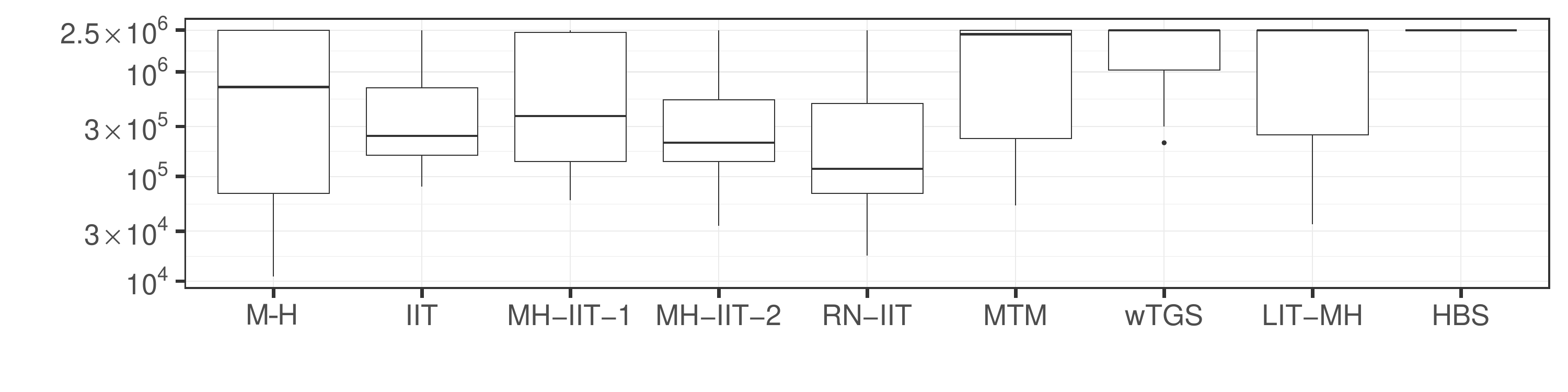}
    \caption{Box plot for the number of evaluations of $\alpha$  (truncated at $2.5 \times 10^6$) needed to find $\hat{\bsgamma}$ in the variable selection simulation with $p=5,000$ and $n=1,000$. Each box represents $100$ replicates.}
    \label{fig:vs.sim}
\end{figure}
   
\section{Complexity Analysis of Algorithm~\ref{alg:mh-iit}} \label{sec:complexity.two}  
\subsection{Theoretical Results}\label{sec:theory-complexity}
The complexity of importance tempering methods consists of two parts: the convergence rate of the importance sampling  estimator $\hat{f}_T$ given by~\eqref{eq:iit.estimator}, and the computational cost for obtaining each $(x^{(i)}, w^{(i)})$.  A general theory on this complexity analysis is provided in Section~\ref{sec:theory} of the \SUPP{}, and in this section we focus on its application to Algorithm~\ref{alg:mh-iit},  MH-IIT.  

The samples $(x^{(1)}, \dots, x^{(T)})$  generated from MH-IIT form a Markov chain, and we denote its transition matrix by $\tP$. Note that $\tP$ depends on the choice of $\alpha$ but not $\rho$.  
We can express $\tP$ and its stationary distribution $\tpi$ by 
\begin{equation}\label{eq:def.tP}
    \tP(x, y) =  \frac{ \alpha(x, y)}{ Z(x)} \ind_{\cN_x}(y), \quad  \tpi(x) 
    = \frac{ \pi(x) Z (x) }{ \pi(Z)}, 
\end{equation}  
where we recall $\alpha(x, y) = h(\pi(y) / \pi(x))$ for some balancing function $h \in [0, 1]$, $Z(x) = \sum_{y \in \cN_x} \alpha(x, y)$, and $\pi(f) = \sum_{x \in \cX} \pi(x) f(x)$.  

Let $K(x)$ be as given in Lemma~\ref{lemma:exp-weight}. 
We propose to measure the complexity of MH-IIT using $\kappa_{ \rho} \sigma^2_{ \rho}(f) $, where $\kappa_{\rho} = \sum_{x \in \cX} \tpi(x) \, \bbE[ K(x) ]$ is the expected cost of each iteration of Algorithm~\ref{alg:mh-iit}, and $\sigma^2_{\rho}(f)$ denotes the asymptotic variance of $\sqrt{T} \hat{f}_T$ as $T \rightarrow \infty$.  
When $\rho \equiv 0$, this is just the asymptotic variance of the time-average estimator produced from a vanilla M--H algorithm; see Section~\ref{sec:theory} in the \SUPP{}. 
When $\rho \equiv c \in [0,1]$, we simply write  $\kappa_c, \sigma^2_{c}$ in place of  $\kappa_\rho, \sigma^2_{\rho}$. 
Using~\eqref{eq:def.tP} and Lemma~\ref{lemma:exp-weight}, it is easy to verify that   
\begin{align*}
    \kappa_{0} = \sum_{x \in \cX}  \frac{ \tpi (x) }{   Z(x) / |\cN_x| } \leq \frac{N_{\mathrm{max}}}{\pi(Z) }, \quad 
      \kappa_{ 1} = \sum_{x \in \cX} \tpi(x) |\cN_x| \leq N_{\mathrm{max}}, 
\end{align*}  
where $N_{\mathrm{max}} = \max_x |\cN_x|$ is the maximum neighborhood size.
We propose to use $\rho(x) =  a / |\cN_x|$ for some fixed $a$.
By Corollary~\ref{coro:cost.bound}, this guarantees that the computational cost of Algorithm~\ref{alg:mh-iit} is comparable to the lower costs of IIT and M--H schemes.   
\begin{corollary}\label{coro:cost.bound}
For $\rho(x) = a / |\cN_x|$ with $ a\in (0,  \min_x |\cN_x|]$,  we have 
$$\kappa_\rho \leq \min \left\{ a^{-1} (a+1) \kappa_{1}, \; (a+1)  \kappa_{0} \right\}.$$  
\end{corollary} 
\begin{proof}
See Section~\ref{app:proof} in the \SUPP{}.
\end{proof}

Our last theorem characterizes $\sigma^2_{\rho}(f)$, where $\Gap(\cdot)$ denotes the spectral gap of a transition matrix  or a transition rate matrix.  

\begin{theorem}\label{coro:mh.iit} 
Fix a balancing function $h$, and suppose that $\Gap( \tP )<1$. 
For $f \colon \cX \rightarrow \bbR$ with $\pi(f) = 0$, define $\gamma(f) = \pi(Z ) \pi( f^2 / Z ). $
For any $\rho \colon \cX \rightarrow [0, 1]$, we have  \smallskip \\ 
(i) $\sigma^2_{1}(f)  \leq \sigma^2_{\rho}(f)  \leq  \sigma^2_{0}(f) \leq \sigma^2_{1}(f) + \gamma(f)$; \smallskip \\
(ii)  there exists $f^*$ with $\pi(f^*) = 0$ such that $\sigma^2_{1}(f^*)  \leq \sigma^2_{\rho}(f^*)  \leq   2  \sigma^2_{1}(f^*)$; \smallskip \\
(iii) $\sigma^2_{ \rho}(f)  \leq 2 \pi( f^2 ) / \Gap (\tP^{\cont}  )$ 
where   the transition rate matrix $\tP^{\cont} $ is defined by 
\begin{align*}
    \tP^{\cont} (x, y) = \frac{ \tP(x, y) Z (x) }{ \pi(Z ) }, \quad \text{ if } x \neq y,  
\end{align*}
and $\tP^{\cont} (x, x) = - \sum_{x' \neq x} \tP^{\cont}  (x, x')$. 
\end{theorem}
\begin{proof}
See Section~\ref{app:proof} in the \SUPP{}.
\end{proof}

\begin{remark}\label{rmk:complexity.approx}
By parts (i) and (ii),  roughly speaking, the choice of $\rho$ does not have a significant impact on the asymptotic variance (at least for some function $f$). 
Hence, by part (iii), we can estimate the complexity of MH-IIT using 
\begin{equation}\label{eq:alg3.complexity}
    \mathrm{Comp}(\rho) = \frac{ \kappa_{\rho} }{\Gap(\tP^{\cont} )}. 
\end{equation} 
The use of a continuous-time Markov chain for studying importance sampling estimators was employed in~\citet{zanella2019scalable} for analyzing their variable selection sampler and in~\citet{zhou2022rapid} for analyzing IIT. 
\end{remark}

\subsection{A Numerical Example}\label{sec:example.complex}
We revisit the toy example  $\pi_{\mathrm{dep}}$ discussed in Section~\ref{sec:sim.toy}. 
Consider 
\begin{equation}\label{eq:def.hc}    
h_c(r) = (1 \wedge r e^{-c} ) \vee (r \wedge e^{-c} ),  \quad c \geq 0. 
\end{equation}    
It is straightforward to check that $h_c$ is  a balancing function and takes values in $[0, 1]$. When $c = 0$, we get $h_c(r) = 1 \wedge r$. As $c \rightarrow \infty$, we get $e^c \, h_c(r)  = 1 \vee r$. Hence, a larger $c$ implies more aggressive behavior of the weighting scheme $\alpha(x, y) = h(\pi(y) / \pi(x))$. 

We fix $p=5$ and numerically calculate the complexity estimate defined in~\eqref{eq:alg3.complexity} for MH-IIT schemes using $h_c$.  To indicate the dependence on the choice of $c$, we write $\mathrm{Comp}(c, \rho)$ instead of $\mathrm{Comp}(\rho)$  and $\Gap(\tP^{\cont}_c )$  instead of $\Gap(\tP^{\cont})$. 
For simplicity, we only consider three choices of $\rho$: $\rho  \equiv 0$, $\rho \equiv 1$, and $\rho  \equiv 0.5$. 
In Table~\ref{tab:gap}, we report the optimal values of $\mathrm{Comp}(c, \rho)$ and $\Gap(\tP^{\cont}_c )$ for each choice of $\rho$. Optimal choices of $c$ are given in parentheses, and the minimum complexity in each column is highlighted in bold. 
When $\rho \equiv 0$ (i.e., the standard M--H algorithm), the optimal choice of $c$ is always $c = 0$, in which case $h_c$ is reduced to $h_c(r) = 1 \wedge r$; this is expected according to Peskun's ordering. 
When $\rho \equiv 1$,  the expected computational cost per iteration is a constant independent of $\alpha$,  and thus the optimal choice of $c$ is the one that maximizes the spectral gap.  
From Table~\ref{tab:gap}, we see that when $\theta = 1$, the smallest complexity is achieved by $\rho \equiv 0$, which is not surprising since $\pi$ is flat. 
As $\theta$ gets larger, the complexity of MH-IIT and IIT becomes significantly smaller than that of M--H. 

\begin{table}
    \centering
        \caption{Optimal complexity estimates for MH-IIT with balancing function $h_c$}
    \begin{tabular}{cccc}
    \hline 
         &   $\theta = 1$   &  $\theta = 2$  &   $\theta = 3$  \\
    \hline
    $\max_{c \geq 0} \Gap(\sP^{\cont}_c )$ &  0.62 ($c=2.43$) & 1.19 ($c=3.53$) &  2.77 ($c=4.58$) \\ 
    \hline
    $\min_{c \geq 0}  \mathrm{Comp}(c, \rho)$ with $\rho \equiv 0$  \hspace{0.5cm} & \textbf{5.19} ($c = 0$) & 5.03 ($c=0$) & 5.0 ($c=0$) \\
    \hline
    $\min_{c \geq 0}  \mathrm{Comp}(c, \rho)$ with $\rho \equiv 1$ \hspace{0.5cm} & 8.07 ($c=2.43$) & 4.20 ($c=3.53$) &  \textbf{1.81} ($c=4.58$)  \\ 
    \hline
    $\min_{c \geq 0}  \mathrm{Comp}(c, \rho)$ with $\rho \equiv 0.5$ \hspace{0.5cm} & 7.82 ($c=1.46$) & \textbf{4.18} ($c=2.15)$ & 1.90 ($c = 3.05$) \\
    \hline
    \end{tabular}
    \label{tab:gap}
\end{table}

\section{Discussion}\label{sec:disc}

The primary motivation of the proposed  methodology is to make efficient use of  ``informed proposals'',  which has gained increasing popularity recently~\citep{zanella2020informed, grathwohl2021oops}.  
Unlike  a random-walk-type proposal used in vanilla M--H schemes, an informed proposal as defined in~\citet{zanella2020informed} assigns a proposal weight to each neighboring state that depends on the stationary probability, and it was shown in~\citet{zanella2020informed} that M--H algorithms with locally balanced weighting schemes have desirable theoretical properties. In a broad sense, MTM can  be viewed as an informed M--H scheme as well where a random subset of neighbors is generated in each iteration. 
For problems such as variable selection,  \citet{zhou2022dimension} showed that, surprisingly, locally balanced  proposals could lead to exceedingly low acceptance rates  that are even worse than uninformed M--H algorithms. 
This work advocates for the adoption of importance tempering as a straightforward and effective solution to this issue. Our simulation studies clearly demonstrate its advantage over the conventional M--H implementation.  

The use of informed proposals is also the key distinction between the proposed methods and existing importance sampling-based MCMC algorithms. 
We explicitly specify the proposal weighting scheme $\alpha$ or balancing function $h$ rather than the stationary distribution of the sampling chain, and the resulting importance weight of Algorithm~\ref{alg:mtm-iit}, $1/Z(x, \cS)$, is ``locally informed'' in the sense that it contains information of $\pi$ in the neighborhood of $x$. 
In contrast, for existing Markov chain importance sampling methods, the stationary distribution of the sampler, which we denote by $\tpi$, 
is often set to  $\tilde{\pi}(x) \propto \pi(x)^a$ for some $a > 0$~\citep{jennison1993discussion, gramacy2010importance} or some pre-specified approximation to $\pi$. 
For example, in Bayesian statistical applications, $\tpi$ can be constructed by using an approximate marginal likelihood function or a different prior distribution~\citep{tan2015honest, vihola2020importance}. 
Other approaches to utilizing importance sampling in MCMC algorithms include the algorithm of~\citet{schuster2020markov} and~\citet{rudolf2020metropolis}, which still runs an M--H chain but makes use of rejected samples via importance sampling, 
and the dynamic weighting method~\citep{liu2001theory, liang2002dynamically}, which uses importance weights to increase acceptance probabilities in M--H algorithms on the fly. 

The pseudo-marginal example presented in Section~\ref{sec:pseudo.iit} illustrates that IIT can often be combined with other MCMC techniques. Such extensions have also been studied in~\cite{rosenthal2021jump}, where they devised rejection-free versions of pseudo-marginal and parallel tempering M--H algorithms. Another interesting algorithm closely related to ours is the ``partial neighbor search'' of~\cite{chen2023sampling}, which considers a fixed partition of $\cN_x$ and calculates the importance weight on only one subset of $\cN_x$.  This algorithm is conceptually similar to our MT-IT and RN-IIT algorithms and can be very useful when $\cX$ naturally admits such a partition. On the other hand, our algorithms are easier to implement and applicable to any space.  
 
Finally, we discuss some limitations of this work. 
In Section~\ref{sec:complexity.two}, we apply the theory developed in Section~\ref{sec:theory} of \SUPP{} to study  MH-IIT. 
Essentially the same arguments can be used to understand the complexity of other importance tempering schemes  where the distribution of the importance weight estimate only depends on the current state $x$. 
For MT-IT, we have a bivariate Markov chain $(x^{(i)}, \cS^{(i)})_{i \geq 1}$, where   the importance weight $1/Z(x^{(i)}, \cS^{(i)})$ is not conditionally independent of $x^{(i-1)}$ given $x^{(i)}$.  
Theoretical analysis of such schemes is much more complicated, and it may be beneficial to utilize techniques developed in~\citet{franks2020importance, andral2022importance}. 
Another important question not investigated in this work is how to choose the parameter $m$ for MT-IT and RN-IIT schemes, which has been discussed in Remark~\ref{rmk:m}. 

In all of our numerical examples, we assume parallel computing is not available and force each algorithm considered to have the same computational cost, and thus informed MCMC schemes,  including the proposed samplers and informed M--H algorithms,  are run for much fewer iterations than uninformed M--H. This is a conservative choice, showing informed schemes in the worst possible light. Though, remarkably, our samplers remain superior to other competing methods under these conditions, a simulation study incorporating parallel implementation for informed schemes would be more relevant to modern computational practices and offer even stronger evidence for the advantages of the proposed algorithms. 

\section*{Acknowledgement}
AS was supported in part by NSERC Discovery Grant (NSERC-DG) program. 
QZ was supported in part by NSF grants  DMS-2311307 and DMS-2245591.  
QZ would like to thank Professor Jeffrey Rosenthal for directing us to  related references and Professors Faming Liang, Florian Maire, Giacomo Zanella, Huiyan Sang, Persi Diaconis, among many others, for the helpful conversations.

\newpage 
\appendix 
\numberwithin{algorithm}{section}
\numberwithin{example}{section}
\numberwithin{assumption}{section}
\numberwithin{proposition}{section}
\numberwithin{theorem}{section}
\numberwithin{lemma}{section}
\numberwithin{remark}{section}
\numberwithin{definition}{section}
\numberwithin{equation}{section}
\begin{center}
    \LARGE{Appendices}
\end{center}

\section{Multiple-try Metropolis Algorithm}



In Algorithm~\ref{alg:mtm}, we recall the multiple-try Metropolis (MTM) algorithm. Note that $\alpha$ must satisfy $ \pi(x) q(y \mid x) \alpha(x, y) =  \pi(y) q(x \mid y) \alpha(y, x).$

\begin{algorithm}\setstretch{1.2}  
\caption{Multiple-try Metropolis (MTM).}\label{alg:mtm}
\begin{algorithmic} 
\Require  Proposal $q$,  weighting function $\alpha$, neighborhood size $m \geq 2$, number of iterations $T \geq 1$, initial state $x^{(0)} \in \cX$.
\For{$i = 0, \dots, T$}
\State Draw $\cS^{(i)} = (y_1, \dots, y_m )$ from $q(\cdot \mid x^{(i)})$.   
\State Calculate $\alpha(x^{(i)}, y)$  for every $y \in \cS^{(i)}$.
\State Calculate $Z(x^{(i)}, \cS^{(i)}) = \sum_{y \in \cS^{(i)}} \alpha(x^{(i)}, y) $.
\State Draw $x' \in \cS^{(i)}$ with probability $ \alpha(x^{(i)}, x')/Z(x^{(i)}, \cS^{(i)})$.
\State Draw $(y'_2, \dots, y'_{m})$ from the conditional distribution $q( \cdot \mid x',   x^{(i)})$.
\State Set $\cS' \leftarrow  \{x^{(i)}, y'_2, \dots, y'_{m}\}$. 
\State Calculate $\alpha(x', y')$  for every $y' \in \cS'$ .
\State Calculate $Z(x', \cS') = \sum_{y' \in \cS'} \alpha(x', y') $. 
\State Calculate the acceptance probability $a  =  \min\{1, Z(x, \cS)/ Z(x', \cS') \}$.
\State With probability $a$, set $x^{(i+1)} \leftarrow x'$; with probability $1 - a$,  set $x^{(i+1)}  \leftarrow x$.
\EndFor 
\Ensure  Samples $x^{(1)}, \dots, x^{(T)}$.
\end{algorithmic}
\end{algorithm}

\newpage 

\section{Complexity  Analysis of Importance Tempering Methods}\label{sec:theory}

\subsection{Measuring Complexity of Importance Tempering Schemes}\label{sec:def.complexity}  

For our theoretical analysis, we consider the following setting where the importance weight estimate only depends on the current state. 
Let $\cX$ be discrete and target distribution $\pi > 0$. 
Let$\bigl((X_i, W_i)\bigr)_{i \in \bbN}$ be a bivariate Markov chain with state space $\cX \times [0, \infty)$ and transition kernel 
\begin{equation}\label{eq:def.git.kernel}
    \sP_{\GIT}( (x, w), (x', \ud w')) = \tsP(x, x') \sR(x', \ud w'), 
\end{equation}
where  $\tsP, \sR$ are transition kernels that satisfy the following conditions. 

\begin{assumption}\label{ass1}
    \rm $\tsP$ is irreducible and reversible w.r.t. a probability distribution $\tpi > 0$, and there exists a constant $\fC > 0$ such that 
    \begin{align*}
        \int_0^\infty   w \, \sR(x, \ud w) =  \frac{\fC \,  \pi(x) }{ \tpi(x)}, \quad \quad \forall\, x \in \cX.
    \end{align*} 
\end{assumption}

Given a function $f$, denote its expectation with respect to $\pi$ by $\pi(f) = \sum_x \pi(x) f(x)$, and 
define the importance sampling estimator for $\pi(f)$ by 
\begin{equation}\label{eq:i-mcmc.hatf} 
\hat{f}_T (\tsP, \sR)  =  \frac{ \sum_{i=1}^T  f(X_i) W_i }{ \sum_{i=1}^T  W_i }, 
\end{equation}    
where $T$ denotes the number of samples. 
Denote the collection of   ``centered'' functions by 
$$L_0^2(\cX, \pi) = \{ f \colon f \text{ is a mapping from $\cX$ to $\bbR$}, \; \pi(f) = 0, \;  \pi(f^2) < \infty \}.$$  
For   $f  \in L_0^2(\cX, \pi)$, by a routine argument that applies the central limit theorem for the bivariate Markov chain $((X_i, W_i))_{i \in \bbN}$ and Slutsky's theorem (see, e.g. Proposition 5 in~\citet{deligiannidis2018ergodic}),  we find that 
$\sqrt{T} \hat{f}_T  \overset{\mathrm{Law}}{\longrightarrow} N(0, \sigma^2(f; \tsP, \sR)),$ 
where $\sigma^2(f; \tsP, \sR)$ is called the asymptotic variance of the estimator $\hat{f}_T$. 
Asymptotic variances are commonly used in the MCMC literature to measure the efficiency of samplers. 

To derive a practically useful complexity metric for importance tempering  schemes, we need to take into account the  computational cost for simulating the kernel $\sP_{\GIT}$, which can vary widely across different algorithms.  Since  $X_{i + 1} \sim \tsP(X_i, \cdot)$ and $W_i \sim \sR(X_i, \cdot)$, we may assume that the computational cost for generating $(X_{i+1}, W_i)$ only depends on the value of $X_i$. 
Hence, we denote by $\kappa(x; \tsP, \sR)$ the cost of simulating a pair $(X, W)$ where $X \sim \tsP(x, )$ and $W \sim \sR(x, \cdot)$. Our formalism allows $\kappa$ to be any function of $x$, but in practice we will typically choose this to be the expected  number of times that the posterior is evaluated (up to a normalizing constant) during each step of the algorithm.
Let 
\begin{equation}\label{eq:def.kappa}
 \kappa(\tsP, \sR ) = \sum_{x \in \cX} \tpi(x) \kappa (x; \tsP, \sR)
\end{equation} 
denote the cost averaged over $\tpi$.  
We use  $\sigma^2(f; \tsP, \sR)   \kappa(\tsP, \sR)$ to measure the complexity of the importance tempering scheme $\sP_{\GIT}$, which we think of as the effective computational cost of each independent sample from $\pi$.  

\begin{example} [Relationship to Theory for Uninformed M--H]\label{ex:comp.mh}
\rm Let $\cX$ satisfy Assumption $1$ and the proposal at $x$ be a uniform distribution on $\cN_x$. Consider the vanilla M--H algorithm with transition matrix given by 
\begin{equation*} \label{eq:mh.kernel}
 \sP_{\mh}  (x, y) =  \ind_{ \{y \neq x\} } \frac{ \alpha(x, y) }{N}    + \ind_{ \{y = x\} } \left( 1 -  \sum_{x' \neq x} \frac{   \alpha(x, x')}{N} \right), 
\end{equation*}   
where $\ind$ is the indicator function. Let $(Y_1, Y_2, \dots, )$ denote a trajectory of the M--H sampler. Then $T^{-1/2} \sum_{i=1}^T f(Y_i)$ converges in distribution to $N(0, \sigma^2_{\mh}(f))$, where $\sigma^2_{\mh}(f)$ denotes the corresponding asymptotic variance. 
If we rewrite it as an importance tempering scheme (i.e., $X_i$ is the $i$-th accepted state and $W_i$ is the number of iterations the M--H sampler stays at $X_i$), then the pair $(\tsP, \sR)$ is given by 
\begin{align*}
 &\tsP_{\mh}  (x, y) = \ind_{ \{y \neq x\} }  \frac{   \alpha(x, y) }{  Z(x) }, \quad   
  \sR_{\mh}(x, \cdot) =  \mathrm{Geom} ( Z(x) / N), 
\end{align*}
where $Z(x) = \sum_{x' \neq x} \alpha(x, x')$. 
Since each uninformed MH iteration requires one evaluation of $\pi$ for the proposed state, we have $\kappa(x; \tsP_{\mh}, \sR_{\mh}) = N / Z(x)$, the expected number of iterations needed to leave $x$. 
Recall that the stationary distribution of $\tsP_{\mh}$ is given by $\tilde \pi(x) = \pi(x)   Z(x) / \pi( Z)$, and it follows that 
$\kappa( \tsP_{\mh}, \sR_{\mh} ) =  N / \pi(Z)$.  
By definition, $\sigma^2( \tsP_{\mh}, \sR_{\mh} )$ is the asymptotic variance of 
\begin{align*}
   \sqrt{N}   \frac{ \sum_{i=1}^N  f(X_i) W_i }{ \sum_{i=1}^N  W_i } 
  =  \frac{ \sum_{i=1}^{T_N}  f(Y_i)  }{\sqrt{T_n}} \frac{1} { \sqrt{ (\sum_{i=1}^N  W_i )/ N} }, 
\end{align*}
where $T_k = \sum_{i=1}^k W_i$ denotes the number of iterations the MH sampler spend on the first $k$ accepted states.  An application of Slutsky's theorem and Law of Large Numbers yields that 
 $   \sigma^2( \tsP_{\mh}, \sR_{\mh} )\kappa( \tsP_{\mh}, \sR_{\mh} ) = \sigma^2_{\mh}(f),$  
which shows that our definition of the complexity metric is consistent with the existing theory on MCMC methods. 
\end{example} 

\subsection{Bounding the Asymptotic Variance}\label{sec:comparing.git} 

Fix an irreducible and reversible transition matrix $\tsP$ with stationary distribution $\tpi$. Let $\cR(\tsP, \pi)$ denote the collection of all possible choices of $\sR \colon \cX \times \cB([0,\infty)) \rightarrow [0,1]$ such that Assumption~\ref{ass1} is satisfied, where $\cB$ denotes the Borel $\sigma$-algebra. 
For each $\sR \in \cR(\tsP, \pi)$, let $m_{\sR}(x) = \int   w  \, \sR(x, \ud w)$.  By Assumption~\ref{ass1}, we find that 
\begin{align*}
   m_{\sR}(x) =  \tpi ( m_{\sR} )   \frac{ \pi(x) }{ \tpi(x) };  
\end{align*} 
that is,  the constant $\fC$ in Assumption~\ref{ass1} is given by $\fC = \tpi ( m_{\sR} )$.  
Let 
\begin{equation}\label{eq:def.v} 
v_{\sR}(x) =  \frac{1}{ \tpi ( m_{\sR} )^2 } \int   \left(w - m_{\sR}(x) \right)^2 \sR(x, \ud w) 
\end{equation} 
denote the variance of the normalized importance weight estimator at  $x$. 
For $f \in L_0^2 (\cX, \pi)$, the asymptotic variance $ \sigma^2(f; \tsP, \sR) $ can be expressed by 
\begin{equation}\label{eq:def.asymp.var}
 \sigma^2(f; \tsP, \sR)  = \lim_{T \rightarrow \infty}  \frac{1}{T \, \tpi ( m_{\sR} )^2} \, \Var \left(   \sum_{i=1}^T  f( X_i ) W_i   \right). 
\end{equation}   
By~\citet[Corollary 21.1.6]{douc2018markov} and the fact that $\sP_{\GIT}$ is Harris recurrent and reversible with respect to  $\tpi(x) \sR(x, \ud w)$ under Assumption~\ref{ass1}, 
\eqref{eq:def.asymp.var} holds regardless of the distribution of $(X_0, W_0)$.
So henceforth whenever we use~\eqref{eq:def.asymp.var}, we assume $(X_0, W_0)$ is drawn from the stationary distribution,   $X_0 \sim \tpi$ and $W_0 \sim \sR(X_0, \cdot)$.  
We first prove a comparison result. 

\begin{theorem}\label{th:comparison}   
For  $\sR_1, \sR_2 \in \cR(\tsP, \pi)$ and $f \in L_0^2(\cX, \pi) $, 
$$\sigma^2(f; \tsP, \sR_1) - \sigma^2(f; \tsP, \sR_2) =  \tpi ( f^2 (v_{\sR_1} - v_{\sR_2}) ),$$ 
where $v_{\sR_j}$ is given by~\eqref{eq:def.v}.    
\end{theorem}
\begin{proof}
See Section~\ref{sec:appx.gap}. 
\end{proof}

\begin{remark}\label{rmk:opt.iit}
\rm   By Theorem~\ref{th:comparison},  if we have an upper bound on $\sigma^2(f; \tsP, \sR)$, we can use it to bound $\sigma^2(f; \tsP, \sR')$ for any other  $\sR' \in \cR(\tsP, \pi)$.  Further, for any $f \in L_0^2(\cX, \pi) $, $\min_{\sR \in \cR(\tsP, \pi)} \sigma^2(f; \tsP, \sR) $ is achieved when for every $x \in \cX$, $\sR(x, \cdot)$ is a Dirac measure; in other words,  exact calculation of importance weights minimizes the asymptotic variance.   
\end{remark} 

Recall the  definition of  spectral gap. 
\begin{definition}\label{def:gap1}
\rm Given a reversible transition kernel or transition matrix $\sP$  with state space $\cY$ and stationary distribution $\mu$,  let $\sP_0$ denote its restriction to the space 
$$L_0^2(\cY, \mu) = \{ f \colon f \text{ is a mapping from $\cY$ to $\bbR$}, \; \mu(f) = 0, \; \mu(f^2) < \infty \},$$
Let $\mathrm{Spec}(\sP_0)$ denote the spectrum of $\sP_0$ and 
$$\Gap(\sP) \coloneqq 1 - \sup\{ \lambda \colon \lambda \in \mathrm{Spec}(\sP_0) \}$$ 
denote the spectral gap of $\sP$.  
\end{definition}
Note that, when $\sP$ is a transition matrix, $\Gap(\sP) = 1 - \lambda_2(\sP)$, where $\lambda_2(\sP)$ is the second largest eigenvalue of $\sP$. 

\begin{definition}\label{def:gap2}
\rm For a transition rate matrix $\sP^{\cont}$ (``$\cont$'' denotes continuous-time) with eigenvalues $0 = \lambda_1( \sP^{\cont} ) \geq \lambda_2( \sP^{\cont} ) \geq \cdots$, we define $\Gap(\sP^{\cont}) = -\lambda_2( \sP^{\cont} )$. 
\end{definition}

We  prove a technical lemma which shows that the spectral gap of the kernel $\sP_{\GIT}$ defined in~\eqref{eq:def.git.kernel} coincides with $\Gap(\tsP)$ whenever $\Gap(\tsP) < 1$. 
\begin{lemma}\label{lm:gap.tsp} 
Suppose Assumption~\ref{ass1} holds for $(\tsP, \sR)$, and define $\sP_{\GIT}$ as in~\eqref{eq:def.git.kernel}. 
If $\Gap(\tsP) \geq 1$, then $\Gap(\sP_{\GIT}) \geq 1$;
if $\Gap(\tsP) < 1$,  then $\Gap(\sP_{\GIT}) = \Gap(\tsP)$. 
\end{lemma}
\begin{proof}
See Section~\ref{sec:appx.gap}. 
\end{proof}
 
Lemma~\ref{lm:gap.tsp} enables us to obtain an upper bound on $\sigma^2(f)$ using $\Gap(\tsP)$. 
Note that for most real problems of interest, the challenge is to find a nonzero lower bound on $\Gap(\tsP)$ (which in turn yields an upper bound on the asymptotic variance), so we will not consider the case $\Gap(\tsP) \geq  1$. 

\begin{proposition}\label{prop:conv.bound1}
Suppose $\Gap(\tsP) < 1$. For any $f \in L_0^2(\cX, \pi)$ and $\sR \in \cR(\tsP, \pi)$, 
\begin{align*}
\sigma^2(f; \tsP, \sR) \leq  \frac{ 2  \left(  \gamma(f) + \tpi ( f^2 v_{\sR} ) \right)}{  \Gap(\tsP)}, 
\end{align*}  
where $\gamma(f) =  \pi(  f^2 \pi / \tpi ).$ 
Further, there   exists  some $f^* \in L_0^2(\cX, \pi)$ such that $\sigma^2(f^*; \tsP, \sR)  \geq \gamma(f^*) / \Gap(\tsP)$. 
\end{proposition}

\begin{proof}
See  Section~\ref{sec:appx.gap}. 
\end{proof}

An alternative way to bound the asymptotic variance is to use Theorem~\ref{th:comparison} and fix a ``reference'' choice of $\sR$. A convenient choice is   $\sR_{\exp}$, which is defined by: 
\begin{align*}
 \sR_{\exp}(x, \cdot) \text{ is the exponential distribution with mean } \fC  \pi(x) / \tpi(x), \text{ where $c>0$ is a constant}. 
\end{align*}
The constant $\fC$ in the above definition can be chosen arbitrarily. For $\sR_{\exp}$, we have the following asymptotic variance, which is derived by converting $\sP_{\GIT}$ into a continuous-time Markov chain. 

\begin{proposition}\label{prop:conv.bound2}
Let $ \sR_{\exp} \in \cR(\tsP, \pi)$ be as given above. 
For  $f \in L_0^2(\cX, \pi)$, 
\begin{align*}
\sigma^2(f; \tsP, \sR_{\exp}) \leq  \frac{ 2 \pi(f^2) }{ \Gap (\sP^{\cont} ) }, 
\end{align*}
where the transition rate matrix $\sP^{\cont}$ is defined by 
\begin{align*}
\sP^{\cont}(x, y) = \left\{\begin{array}{cc}
\tsP(x, y) \tpi(x) / \pi(x),  & \text{ if } x \neq y, \\
- \sum_{x' \neq x} \sP^{\cont}(x, x'),  & \text{ if } x = y. 
\end{array} \right.
\end{align*}
\end{proposition}

\begin{proof}
See Section~\ref{sec:appx.gap}.  
\end{proof}

\begin{remark}\label{rmk:compare.exp}
\rm By Theorem~\ref{th:comparison}, for  any $\sR \in \cR(\tsP, \pi)$ such that $v_{\sR}(x) \leq \pi(x)^2 / \tpi(x)^2$ for every $x \in \cX$, 
$\sigma^2(f; \tsP, \sR)  \leq \sigma^2(f; \tsP, \sR_{\exp}) $, so the bound in Proposition~\ref{prop:conv.bound2} can be applied as well.
\end{remark}

\subsection{Proofs} \label{sec:appx.gap} 

\begin{proof}[Proof of Theorem~\ref{th:comparison}] 
Let $(X_i)_{i \in \bbN}$ be a Markov chain with transition matrix $ \tsP$ and $X_0 \sim \tpi$. 
Let $W_i^{(j)}   \sim \sR_j(X_i, \cdot)$ and $\widecheck{W}_i^{(j)} = W_i^{(j)} / \tpi ( m_{\sR_j} )$.  
Note that for $j = 1, 2$,  $\bbE[ \widecheck{W}_0^{(j)} \mid X_0 ] = \pi(X_0) / \tpi(X_0)$. By conditioning on $X_0$, we get  
\begin{align}
\bbE[ f( X_0 ) \widecheck{W}_0^{(j)}  ] =  \bbE[  \, \bbE[ f( X_0 ) \widecheck{W}_0^{(j)} \mid X_0 ]  \, ]  = \pi( f ) = 0,  \label{eq:cond1} \\
\bbE[ ( f( X_0 ) \widecheck{W}_0^{(j)} )^2  ] = \bbE \left[  f( X_0 )^2  \left (  \frac{\pi(X_0)^2 }{\tpi(X_0)^2} + v_{\sR_j}(X_0)  \right) \right].   \label{eq:cond2} 
\end{align} 
By~\eqref{eq:def.asymp.var} and~\eqref{eq:cond1}, 
\begin{align*}
\sigma^2(f; \tsP, \sR_j) 
= \lim_{T \rightarrow \infty}  \frac{1}{T }  \bbE \left[   \left(  \sum_{i=1}^T  f( X_i ) \widecheck{W}_i^{(j)}   \right)^2 \right]. 
\end{align*} 
For any $T \geq 1$,  using the conditioning argument again, we find that 
\begin{align*} 
 & \bbE  \left[ \left(   \sum_{i=1}^T  f( X_i ) \widecheck{W}_i^{(1)}   \right)^2 - \left(   \sum_{i=1}^T  f( X_i ) \widecheck{W}_i^{(2)}     \right)^2   \right] \\
 = \;& \sum_{i=1}^T  \bbE  \left[     \left( f( X_i ) \widecheck{W}_i^{(1)}   \right)^2 -     \left( f( X_i ) \widecheck{W}_i^{(2)}   \right)^2   \right] \\
 = \;&  T \, \bbE  \left[     \left( f( X_0 ) \widecheck{W}_0^{(1)}     \right)^2 -     \left( f( X_0 ) \widecheck{W}_0^{(2)}   \right)^2   \right] \\
= \;& T \,  \tpi ( f^2 (v_{\sR_1} - v_{\sR_2}) ), 
\end{align*}
where the last step follows from~\eqref{eq:cond2}.   The asserted result then follows. 
\end{proof} 

\begin{proof}[Proof of Lemma~\ref{lm:gap.tsp}]
Clearly, $\sP_{\GIT}$ is reversible with respect to the distribution $\tpi( \ud x) \sR(x, \ud w)$. 
Let $((X_i, W_i))_{i \in \bbN}$ be a Markov chain with kernel $\sP_{\GIT}$, $X_0 \sim \tpi$ and $W_0 \sim \sR(X_0, \cdot)$. 
We can express  $\Gap(\sP_{\GIT})$ by~\citet[Theorem 22.A.19]{douc2018markov}
\begin{equation}\label{eq:gap}
\Gap(\sP_{\GIT}) = \inf_{f \, :\,  \Var ( f(X_0, W_0)  ) =1} (\bbE [  f(X_0, W_0)^2]  -  \bbE [ f(X_0, W_0) f(X_1, W_1)  ]). 
\end{equation}
Similarly, 
\begin{equation}
\Gap(\tsP)  = \inf_{g \, : \, \Var(g(X_0) ) =1}  (\bbE[ g(X_0)^2 ] -  \bbE[ g(X_0) g(X_1) ]).  
\end{equation}
Since any function $g(x)$ can also be seen as  a bivariate function $f$ such that $f(x, w) = g(x)$, we have $\Gap(\sP_{\GIT}) \leq \Gap(\tsP)$. 

For each $x \in \cX$ and function $f(x, w)$,  define $m_f(x) = \int f(x, w) \sR(x, \ud w)$, and $v_f(x) = \int  ( f(x, w)   - m_f(x) )^2 \sR(x, \ud w) $. 
 By conditioning and using the fact that $W_i$ is conditionally independent of everything else given $X_i$, we find that 
\begin{align*}
 \Var ( f(X_0, W_0) )  = \;&   \bbE[  v_f(X_0)  ] +  \Var(m_f(X_0) ), \\ 
   \bbE [ f(X_0, W_0) f(X_1, W_1)  ]  =\;& \bbE[ m_f(X_0) m_f(X_1) ], \\
  \bbE [  f(X_0, W_0)^2] =\;&   
  \bbE[  v_f(X_0)  ] + \bbE[ m_f(X_0)^2 ]. 
\end{align*} 
For real numbers $a, b, c \geq 0$, we have $(a + b)/(a + c) \geq b/c$ if $b \leq c$, and  $(a + b)/(a + c) \geq 1$ if $b \geq c$. Hence,  
\begin{align*}
\Gap(\sP_{\GIT}) = \;&  \min_f   \frac{ \bbE[  v_f(X_0)  ] +   \bbE[ m_f(X_0)^2 ] -  \bbE[ m_f(X_0) m_f(X_1) ]  }{     \bbE[  v_f(X_0)  ] +  \Var(m_f(X_0) ) } \\
\geq \;&  \min \left\{ 1, \, \min_f   \frac{    \bbE[ m_f(X_0)^2 ] -  \bbE[ m_f(X_0) m_f(X_1) ]  }{     \Var(m_f(X_0) ) } \right\} 
= \min \left\{ 1, \,  \Gap(\tsP)   \right\}.  
\end{align*}  
Hence,  if $\Gap(\tsP) \geq 1$, then $\Gap(\sP_{\GIT}) \geq 1$. 
If $\Gap(\tsP) < 1$, then $\Gap(\sP_{\GIT}) \geq \Gap(\tsP)$. Since we have already shown $\Gap(\sP_{\GIT}) \leq \Gap(\tsP)$, the two spectral gaps must coincide.  
\end{proof}

\begin{proof}[Proof of Proposition~\ref{prop:conv.bound1}]
Let $(X_i)_{i \in \bbN}$ be a Markov chain with transition matrix $ \tsP$ and $X_0 \sim \tpi$. Let $W_i \sim \sR(X_i, \cdot)$. 
Applying Proposition 22.5.1 in~\citet{douc2018markov} to the bivariate Markov chain $((X_i, W_i))_{i \in \bbN}$, we get 
\begin{align*}
T^{-1} \, \Var \left(   \frac{ \sum_{i=1}^T  f( X_i ) W_i  }{  \tpi( m_{\sR} ) } \right) \leq \frac{2}{ \Gap(\sP_{\GIT}) } \bbE\left[ \frac{f( X_0 )^2 W_0^2 }{  \tpi( m_{\sR} )  ^2}   \right]. 
\end{align*}
By Lemma~\ref{lm:gap.tsp}, $\Gap(\sP_{\GIT}) = \Gap(\tsP)$. The asserted upper bound on $\sigma^2(f; \tsP, \sR)$ then follows from~\eqref{eq:cond2}.  

For the second claim,  recall that $\Gap(\tsP) =  1 -  \lambda_2(\tsP)$, where $\lambda_2(\tsP)$ is the second largest eigenvalue of $\tsP$. 
Let  $g$ be an eigenvector of $\tsP$ such that $\tsP g = \lambda_2 g$, which must satisfy $\tpi(g) = 0$.  By~\citet[Chap. 11.2.3]{bremaud2013markov},  
\begin{align*}
 \lim_{T \rightarrow \infty}  T^{-1} \, \Var \left(   \sum_{i=1}^T  g(X_i) \right) = \frac{1 + \lambda_2 (\tsP)}{1 - \lambda_2 (\tsP)} \tpi( g^2). 
\end{align*}
Define $f^*(x) = g(x) \tpi(x) / \pi(x)$. By Remark~\ref{rmk:opt.iit}, $\sigma^2(f^*; \tsP, \sR_{\dirac} ) $ is minimized when $\sR(x, \cdot)$ is a Dirac measure for every $x$. Let $\sR_{\dirac}  \in \cR(\tsP, \tpi)$ be such a choice of $\sR$. 
Then,  
\begin{align*}
\sigma^2(f^*; \tsP, \sR)  \geq 
\sigma^2(f^*; \tsP, \sR_{\dirac} ) 
=\;&  \lim_{T \rightarrow \infty}  T^{-1} \, \Var \left(   \sum_{i=1}^T  g(X_i) \right) \geq  \frac{\tpi( g^2) }{ \Gap(\tsP)} . 
\end{align*}
But $\tpi( g^2) = \tpi(  (f^*)^2 \pi^2 / \tpi^2 )  = \gamma(f^*)$.  
\end{proof}

\begin{proof}[Proof of Proposition~\ref{prop:conv.bound2}] 
By~\eqref{eq:def.asymp.var}, 
\begin{align*}
\sigma^2(f; \tsP, \sR_{\exp})  = \lim_{T \rightarrow \infty } \frac{1}{T} \Var\left( \sum_{i=1}^T f(X_i) W_i^* \right)
\end{align*}
where  $(X_i)_{i \in \bbN}$ is a Markov chain with transition matrix $ \tsP$ and $X_0 \sim \tpi$, and given $X_i$, $W_i^*$ follows an exponential distribution with mean $\pi(X_i) / \tpi(X_i)$. 
If $\tsP(x, x) = 0$ for all $x \in \cX$, it is easy to see that $\sum_{i=1}^T f(X_i) W_i^* / T$ is   the time average of a continuous-time Markov chain with transition rate matrix $\sP^{\cont}$. 
If $\tsP(x, x) \neq 0$ for some $x$, this interpretation still holds, and one can prove it using the fact that a geometric sum of i.i.d. exponential random variables is still exponential. 
The asserted bound then follows from Proposition 4.29 of~\citet{aldous2002reversible}. 
\end{proof}

\clearpage 
\newpage 

\section{Proofs for the Main Text}\label{app:proof}

\begin{proof}[Proof of Lemma~\ref{lm:mtit}]
Let $\cS = \{x', y_2, \dots, y_m\}$ and $\cS' = \{x, y'_2, \dots, y'_m\}$. By~\eqref{eq:mt.transition}, 
    \begin{align*}
    & \pi(x) \, Z(x, \cS) \, q(x', y_2, \dots, y_m \mid x) \, p_{\mt}( (x, \cS), (x', \cS') )  \\ 
   = \;& \pi(x) \, Z(x, \cS) \, q(x' \mid x) \, q(  y_2, \dots, y_m \mid x,  x') \, p_{\mt}( (x, \cS), (x', \cS') ), \\  
   = \;& (m-1)! \, \pi(x) \, q(x' \mid x)  \, \alpha(x, x') \,  q(  y_2, \dots, y_m \mid x,   x')  \,  q ( y'_2, \dots, y'_m \mid x',   x).  
\end{align*} 
Since $\alpha$ must satisfy~\eqref{eq:alpha}, $\pi(x) q(x' \mid x)  \alpha(x, x')$ is symmetric in $x, x'$. It then follows that $p_{\mt}$ satisfies the detailed balance condition with respect to $\pi_{\mt}$, proving the claim. 
\end{proof}

\begin{proof}[Proof of Lemma~\ref{lm:pm}]
Denote $v = (u_x, u_1, \cdots, u_N)$ and $v' = (u_{x'}, u_1', \cdots, u_{N'})$. Recall that we assume $x' = y_1$, $x = y_1'$, $u_x = u_1'$, and $u_1 = u_{x'}$. To prove the claim, observe that
\begin{align*}
& \pi_{\pseu}(x,  v) p_{\pseu}((x,  v), (x',   v')) \\
\propto\;& u_x \pi(x) g(u_x|x) Z(x,  v) \left\{ \prod_{i=1}^N g(u_i | y_i) \right\} \frac{\hat \alpha(x, y_1; u_x, u_1)}{Z(x, v)} \left\{ \prod_{i=2}^{N'} g(u_i' | y_i') \right\} \\
=&  u_x \pi(x)  \hat \alpha(x, y_1; u_x, u_1) 
 \left\{\prod_{i=1}^N g(u_i | y_i)\right\} \left\{\prod_{i=1}^N g(u_i' | y_i')\right\}.
\end{align*}
Since $\hat \alpha(x, y_1; u_x, u_1)  $ can be expressed by 
\begin{equation}
    \hat \alpha(x, y_1; u_x, u_1)  
    = \hat \alpha(x, x'; u_x, u_{x'})  
    = h \left( \frac{\hat{\pi}(x')}{\hat{\pi}(x)} \right)
    = h \left( \frac{ u_{x'} \pi(x')}{u_x \pi(x)} \right), 
\end{equation}
and $h$ is a balancing function, we get 
 $   u_x \pi(x)  \hat \alpha(x, x'; u_x, u_{x'})   
=  u_{x'} \pi(x') \hat \alpha(x', x; u_{x'}, u_{x}).$ 
It follows that $p_{\pseu}$ satisfies the detailed balance condition with respect to $\pi_{\pseu}$, thereby proving the claim.   
\end{proof}

\begin{proof}[Proof of Lemma~\ref{lemma:mh-iit}]
Suppose $y \in \cN_x$ (which implies $y \neq x$). The transition probability of the MH-IIT algorithm moving from state $x$ to state $y$ is given by:
\[\tilde{p}(x, y) = \rho \frac{\alpha(x, y)}{Z(x)} + (1 - \rho)\frac{\alpha(x, y)  }{\sum_{x'\in\cN_x}\alpha(x, x') } = \frac{\alpha(x, y)}{Z(x)}.\]
Given that $\alpha(x, y) = h\bigl(\pi(y) / \pi(x)\bigr)$ where $h(\cdot)$ is a balancing function, it follows that 
\[\pi(x)Z(x)\tilde p(x, y) = \pi(x)h\left( \frac{\pi(y) }{\pi(x)} \right)\]
is symmetric in $x$ and $y$. Consequently, $\tilde{p}$ satisfies the detailed balance condition with respect to $\tilde{\pi}$, thereby proving the claim.
\end{proof}

\begin{proof}[Proof of Lemma~\ref{lemma:exp-weight}]
Let $p_Z \coloneqq Z / N$ be the probability that an MH update in Algorithm~\ref{alg:mix.weight} is accepted. Whenever we perform an IIT update, we can think of $p_Z$ as the expected number of remaining MH updates needed to stop (if IIT updates are not allowed). Since geometric distribution is memoryless, this shows that $\bbE[W] = 1 / Z$. 

For a  more straightforward proof, let $I$ denote the number of iterations in Algorithm~\ref{alg:mix.weight} and observe that the first $ I -1$ iterations have to be MH updates. Further, 
\begin{align*}
& \bbP(I  = k  \text{ and the last iteration is an MH update} ) = (1 - \rho )^{k} (1 - p_Z)^{k-1} p_Z  , \\ 
& \bbP(I  = k  \text{ and the last iteration is an IIT update} ) = (1 - \rho )^{k-1} (1 - p_Z)^{k-1} \rho.  
\end{align*}
Thus, the expectation of $W$ can be calculated by  
\begin{align*}
\bbE[ W ] =\;&  \frac{1}{N}\sum_{k = 1}^{\infty} \left\{ (1 - \rho)^{k-1}(1 - p_Z)^{k-1}\rho \left(p_Z^{-1} + k - 1\right)  
 +  (1 - \rho)^k (1 - p_Z)^{k-1} p_Z k \right\} \\
=\;& \frac{1}{N}\sum_{k=1}^{\infty}(1 - \rho)^{k-1}(1 - p_Z)^{k - 1}\biggl\{\rho\left(p_Z^{-1} - 1\right)  
 + \bigl(\rho  + (1 - \rho) p_Z \bigr)k\biggr\} \\
=\;&\frac{1}{N}\left\{\frac{\rho p_Z^{-1} - \rho}{1 - (1 - \rho)(1 - p_Z)} + \frac{\rho + p_Z - \rho p_Z}{[1 - (1-\rho)(1 - p_Z )]^2}\right\}\\
=\;& 1 / Z,
\end{align*}
where in the third step we have used  the identity, 
\begin{equation}\label{eq:sum-fact}
\sum_{k=1}^{\infty}k c^{k-1} = \frac{\rm{d}}{\rm{d} c}\sum_{k=1}^{\infty}c^k =\frac{1}{(1 - c)^2}, \quad \quad \forall \, c \in [0, 1).  
\end{equation}

Similarly, 
\begin{align*}
\Var( W  ) =\;& \sum_{k = 1}^{\infty}\biggl\{ (1 - \rho)^{k-1}(1 - p_Z)^{k-1}\rho\bigl((k - 1) / N\bigr)^2 + (1 - \rho)^k (1 - p_Z)^{k-1} p_Z \left(k / N - Z^{-1}\right)^2\biggr\}\\
=\;& \frac{1}{N^2} \sum_{k = 1}^{\infty}\biggl\{ (1 - \rho)^{k-1}(1 - p_Z)^{k-1}\rho (k - 1)^2 + (1 - \rho)^k (1 - p_Z)^{k-1} p_Z \left(k - p_Z^{-1}\right)^2\biggr\}\\
=\;& \frac{1}{N^2}\sum_{k=1}^{\infty}(1-\rho)^{k-1}(1 - p_Z)^{k-1}\left\{[\rho + (1-\rho)p_Z]k^2 - 2k + \rho + (1-\rho)p_Z^{-1}\right\}\\
=\;&\frac{(1 - p_Z)(1 - \rho)}{Z^2 + Z\rho(N - Z)},
\end{align*}
where  the last step follows from~\eqref{eq:sum-fact} and $\sum_{k=1}^{\infty}(k+1)k c^{k-1} = 2 / (1-c)^3$ for $c \in [0, 1)$.  

For $\bbE[K]$, by the same argument we get 
\begin{align*}
\bbE[K]  =\;& \sum_{k = 1}^{\infty} (1 - \rho)^{k-1}(1 - p_Z)^{k-1}\rho \left(N + k - 1\right) + \sum_{k=1}^{\infty} (1 - \rho)^k (1 - p_Z)^{k-1} p_Z k\\
=\;&\frac{\rho (N - 1)}{1 - (1 - \rho)(1 - p_Z)} + \frac{\rho + p_Z - \rho p_Z}{\bigl(1 - (1-\rho)(1 - p_Z)\bigr)^2}\\
=\;&\frac{\rho(N - 1) + 1}{\rho(1 - p_Z) + p_Z}, 
\end{align*}
which completes the proof. 
\end{proof}

\begin{proof}[Proof of Corollary~\ref{coro:cost.bound}]
Using Lemma~\ref{lemma:exp-weight}, for any $x$, we obtain 
\begin{align*}
\bbE[K(x)] = \frac{(a + 1)|\cN_x| - a}{a(1 - Z(x) / |\cN_x|) + Z(x)}.
\end{align*}
Hence, 
\begin{align*}
\bbE[K(x)] & \leq (a + 1) \frac{|\cN_x|}{Z(x)}, \text{ and } 
\bbE[K(x)]  \leq \frac{(a+1) |\cN_x|}{a + (1 - a / |\cN_x|)Z(x)} \leq \frac{a + 1}{a}|\cN_x|,
\end{align*}
where we have used $a \leq \min_x |\cN_x|$ in the second inequality. This yields 
\begin{align*}
\kappa_{\rho} & \leq \sum_{x\in\cX} \tilde \pi(x) (a + 1) \frac{|\cN_x|}{Z(x)} = (a+1)\kappa_0,\\
\kappa_{\rho} & \leq \sum_{x\in\cX} \tilde \pi(x) \frac{a + 1}{a}|\cN_x| = \frac{a+1}{a}\kappa_1.
\end{align*}
The asserted result then follows. 
\end{proof}

\begin{proof}[Proof of Theorem~\ref{coro:mh.iit}]
We use the notation introduced in Section~\ref{sec:theory} in \SUPP. Let $\sR_{\rho}$ be the transition kernel for generating importance weights used in MH-IIT; that is, $\sR_{\rho}(x, \cdot)$  is the distribution of the output $w$ from Algorithm~\ref{alg:mix.weight} with input state $x$. Since $\tilde \sP$ is reversible with respect to $\tilde \pi \propto \pi Z$, Lemma~\ref{lemma:exp-weight} shows that $\sR_{\rho} \in \cR(\tilde \sP, \pi)$ for any choice of $\rho$ and thus $(\tilde \sP, \sR_{\rho})$ satisfies Assumption~\ref{ass1}. Using the notation of Section~\ref{sec:theory} in \SUPP, we can write $\sigma^2_{\rho} = \sigma^2(f; \tilde \sP, \sR_{\rho})$.  
Let $m_\rho = m_{\sR}$, $v_\rho = v_{\sR}$ with $\sR = \sR_{\rho}$. That is, 
\begin{align}
    m_{\rho}(x) = \int   w  \, \sR_\rho(x, \ud w),  \quad v_{\rho}(x) =  \frac{1}{ \tpi ( m_{\rho} )^2 } \int   \left(w - m_{\rho}(x) \right)^2 \sR_\rho(x, \ud w) 
\end{align} 
By Lemma~\ref{lemma:exp-weight},  we have 
\begin{align*}
& m_\rho(x) = \frac{1}{Z(x)}, \\
& \tilde \pi( m_\rho) = \sum_x  \frac{ \tilde \pi(x) }{Z(x)} = \frac{1}{ \pi(Z) }, \\
& v_\rho(x) = \frac{\pi(Z)^2 (1 - Z(x) / |\cN_x|)(1 - \rho(x))}{Z(x)^2 + \rho(x)Z(x)(|\cN_x| - Z(x))}.
\end{align*}

For every $x$, $v_\rho(x)$ is monotone decreasing in $\rho(x) \in [0, 1]$.  Hence, 
\begin{equation*}\label{eq:Z0}
0 = v_1(x) \leq    v_\rho(x)  \leq   v_0(x), 
\end{equation*} 
where $v_1, v_0$ correspond to $\rho \equiv 1$ and $\rho \equiv 0$, respectively. By Theorem~\ref{th:comparison}, we get 
\begin{align*}
\sigma^2_{1}(f) =  \sigma^2(f; \tilde \sP, \sR_{1}) \leq \sigma^2(f; \tilde \sP, \sR_{\rho})   \leq \sigma^2(f; \sP, \sR_{0}) = \sigma^2_{0}(f) . 
\end{align*} 
A direct calculation gives 
\begin{align*}
    v_0(x) =  \frac{\pi(Z)^2 (1 - Z(x) / |\cN_x| ) }{Z(x)^2 } \leq \frac{ \pi(Z)^2   }{ Z(x)^2 }. 
\end{align*}

Using Theorem~\ref{th:comparison} again, we find that 
\begin{align*}
\sigma^2_{h, 0}(f) - \sigma^2_{h, 1}(f) =\;& \tilde \pi ( f^2 v_0)  
\leq  \sum_{x} f^2(x) \tilde \pi(x)  \frac{ \pi(Z)^2 }{ Z(x)^2 } =  \pi(Z)    \pi\left(  \frac{ f^2  }{Z}  \right),  
\end{align*}  
where we have used $\tilde \pi = \pi Z  / \pi(Z)$. 
This proves part (i).

To prove part (ii), we first use $\tilde \pi = \pi Z  / \pi(Z)$ again to obtain 
$$  \pi(Z)    \pi\left(  \frac{ f^2  }{ Z  }  \right) = \pi\left( \frac{f^2 \pi}{ \tilde \pi} \right),$$ 
which verifies that the definition of $\gamma(f)$ is the same as in Proposition~\ref{prop:conv.bound1}.  
Since $v_1 \equiv 0$, by Proposition~\ref{prop:conv.bound1}, we know that 
there exists $f^*$ such that 
\begin{align*} 
\frac{  \gamma(f^*)}{  \Gap( \tilde \sP)} \leq \sigma^2_{1}(f^*)   \leq  \frac{ 2 \gamma(f^*)}{  \Gap( \tilde \sP)}.
\end{align*}
By part (i), $    \sigma^2_{1}(f^*)   \leq \sigma^2_{\rho}(f^*)   \leq   \sigma^2_{0}(f^*)  \leq \sigma^2_{1}(f^*) + \gamma(f^*)$.
Since $\Gap(\tilde \sP) < 1$, we obtain that 
$\sigma^2_{\rho}(f^*)  \leq2 \sigma^2_{1}(f^*)$. 
This proves part (ii).  

By Remark~\ref{rmk:compare.exp}, for any $\rho$, we have $\sigma^2(f; \tilde \sP, \sR_{\rho}) \leq \sigma^2(f; \tilde \sP, \sR_{\rm{exp}})$ where $\sR_{\rm{exp}}$ is the kernel such that $\sR_{\rm{exp}}(x, \cdot)$ is an exponential distribution for every $x$. Part (iii) then follows from Proposition~\ref{prop:conv.bound2}. 
\end{proof}

\clearpage 
\newpage

\section{Simulation Studies for Section~\ref{sec:mtm}}\label{app:mt-it}
\subsection{Additional   Results for MT-IT with Multivariate Normal Targets}\label{app:mt-it-normal}
For the multivariate normal example discussed in Section~\ref{sec:mtm-sim}, we present additional results for estimation  using either the last $10\%$ of the samples or the entire dataset in Figure~\ref{fig:mtm.normal.supp}. This figure includes box plots similar to those of Figure~\ref{fig:mtm.normal}, comparing estimates of $M_{2,p}$ derived from different sample sizes. The estimates obtained from the last $10\%$ of the samples closely align with those shown in Figure~\ref{fig:mtm.normal}, where the last $50\%$ of the samples were used. When utilizing the full MCMC trajectories, only the MT-IT method with $h(r) = \sqrt{r}$ produces unbiased estimates. The biases observed in the four MTM methods are notably substantial.

\begin{figure}[!h]
    \centering
    \includegraphics[width=0.75\linewidth]{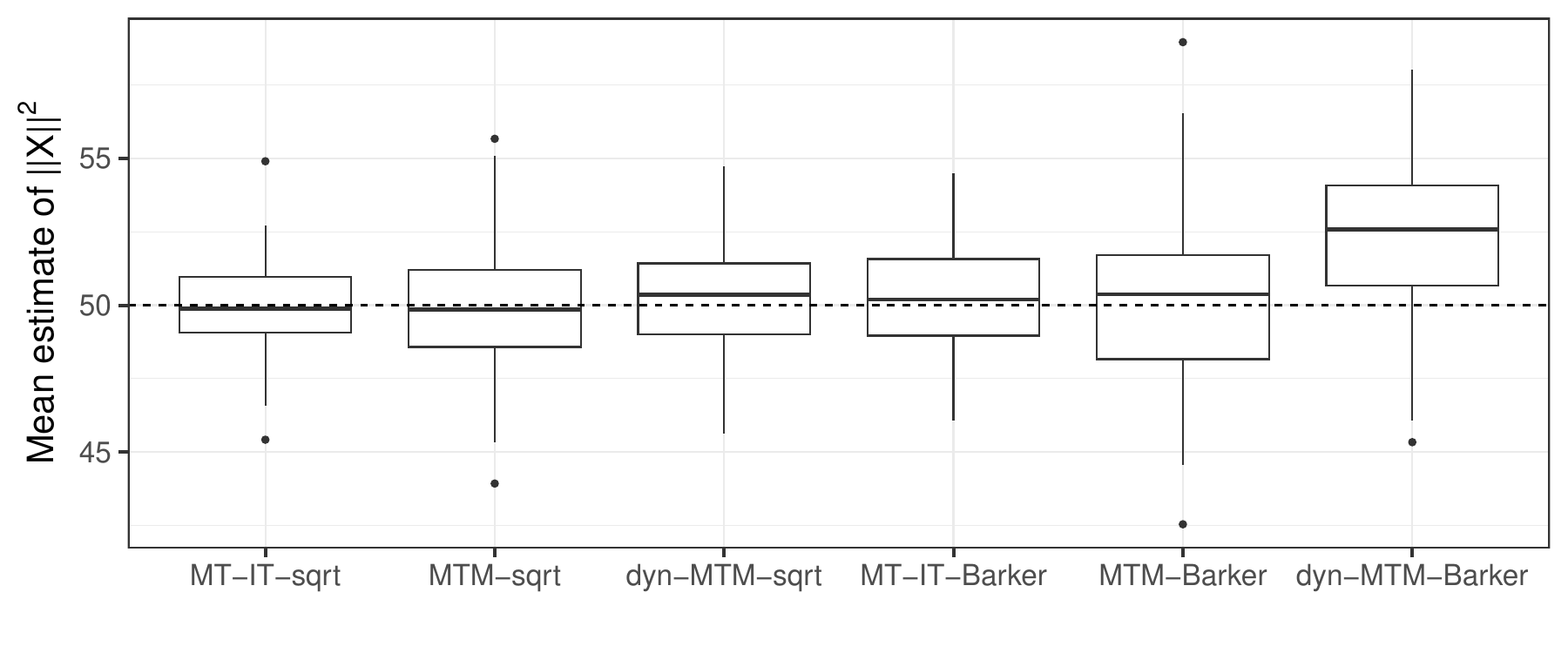}\\
    \includegraphics[width=0.75\linewidth]{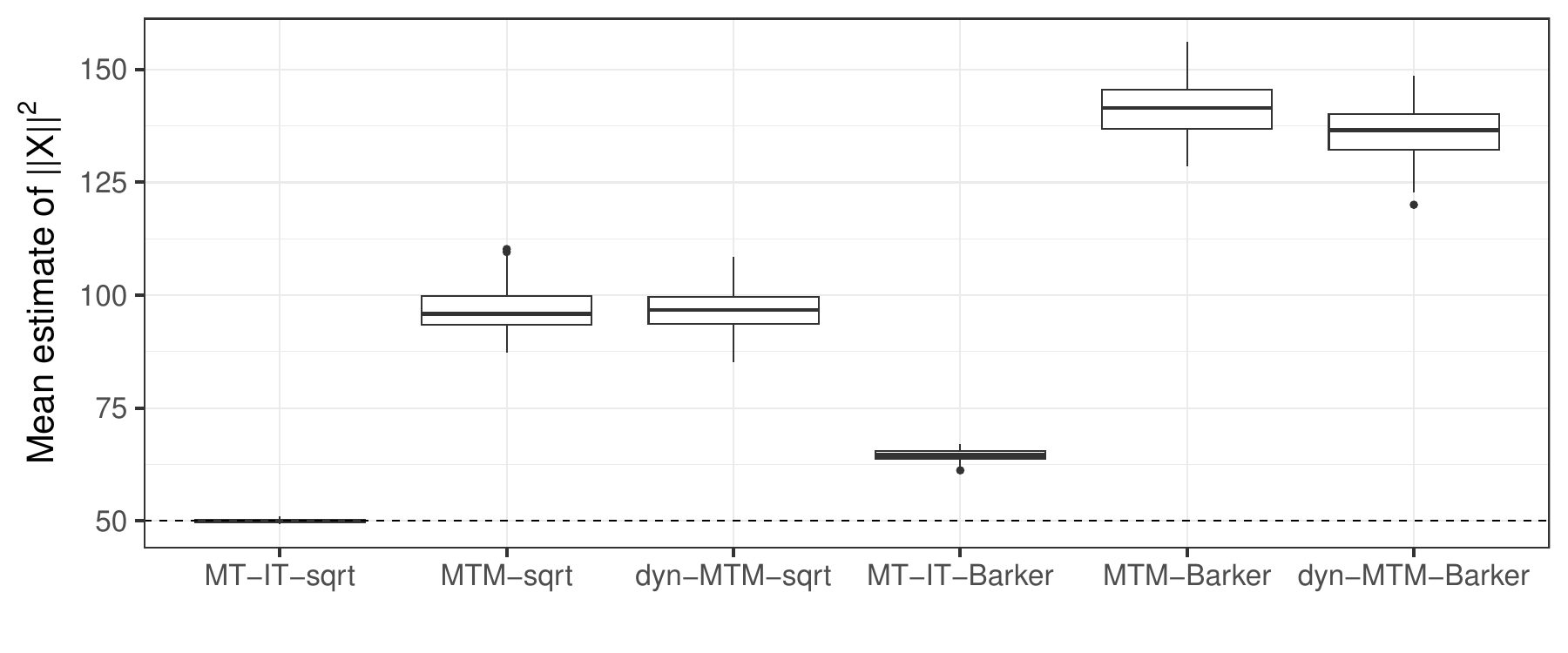}\\
    \caption{ 
    Box plot of the estimate of $M_{2, p}$ over 100 runs with $p = 50$.  
    Top: estimates are calculated by using the last 10\% samples.
    Bottom: estimates are calculated by using all samples. 
    MT-IT: Algorithm~\ref{alg:mtm-iit} with fixed proposal; MTM: MTM with fixed proposal~\citep{gagnon2022improving}; dyn-MTM: adaptive MTM used in~\citep{gagnon2022improving}; sqrt: $h(r) = \sqrt{r}$; Barker: $h(r) = r / (1 + r)$, where $h$ is the balancing function introduced in Remark 1. 
    }\label{fig:mtm.normal.supp} 
\end{figure}
 
\subsection{A Real Data Example on Record Linkage}\label{sec:real} 
In this section, we consider the application of the MT-IT algorithm to a bipartite record linkage problem studied in \citet{zanella2020informed}. 
Let $\bA \in \bbR^{n_1 \times K}$ and $\bB \in \bbR^{n_2 \times K}$ denote two data sets. Let $\ba_i$ denote the $i$-th row of $\bA$, representing the $i$-th record in the first data set, and let $\bb_i$ denote the $i$-th row of $\bB$, representing the $i$-th record in the second data set. The same $K$ covariates are observed for both data sets. 
Let $M_i = j$ if $\ba_i$ is matched with $\bb_j$, and let $M_i = 0$ if $\ba_i$ is not matched with any record in the second data set (``matching'' means that we believe the two records are duplicates). 
Assume that any record is matched with at most one record in the other data set. Hence, we can define $M^{-1}_j = i$ if $M_i = j$, and let $M^{-1}_j = 0$ if $\bb_j$ is not matched with any record in the first data set. The parameter of interest is the $n_1$-dimensional vector $\bM = (M_1, \dots, M_{n_1}) \in \{0, 1, \dots, n_2\}^{n_1}$. 

Fix an arbitrary $\bM$ such that each record in one data set is matched with at most one record in the other data set. We construct a neighborhood $\cN_{\bM} = \{ \bM^{ij} \colon i = 1, \dots, n_1, \; j = 1, \dots, n_2 \}$ with cardinality $n_1 n_2$, where $\bM^{ij}$ is obtained from $\bM$ by one of the following moves:
\begin{itemize}
    \item Add move: If $M_i = 0$ and $M_j^{-1} = 0$, we set $M^{ij}_i = j$.
    \item Delete move: If $M_i = j$, we set $M^{ij}_i = 0$.
    \item  Single switch move (first type): If $M_i = 0$ and $M_j^{-1} = i'$ for some $i' \neq i$, we set $M^{ij}_i = j$ and $M^{ij}_{i'} = 0$.
    \item Single switch move (second type): If $M_i = j'$ for some $j' \neq j$ and $M_j^{-1} = 0$, we set $M^{ij}_i = j$.
    \item Double switch move: If $M_i = j'$ for some $j' \neq j$ and $M_j^{-1} = i'$ for some $i' \neq i$, we set $M^{ij}_i = j$ and $M^{ij}_{i'} = j'$.
\end{itemize}
It is not difficult to check that each pair $(i, j)$ yields a unique $\bM^{ij}$, and thus $|\cN_{\bM}| = n_1 n_2$. 

The goal is to generate samples from the joint posterior distribution $\pi(\bM, p_{\mathrm{match}}, \lambda)$ considered in \citet{zanella2020informed}, where $p_{\mathrm{match}}$ and $\lambda$ are two continuous hyperparameters. A standard approach is to use the Metropolis-within-Gibbs sampler, which alternates between direct sampling from $\pi(p_{\mathrm{match}}, \lambda \mid \bM)$ and an M--H scheme targeting $\pi(\bM \mid p_{\mathrm{match}}, \lambda)$. 

We implement two methods for updating $\bM$: uninformed M--H and the blocked informed M--H scheme of \citet{zanella2020informed} with the number of tries set to 100. We refer to the two resulting Metropolis-within-Gibbs samplers as uninformed and informed M--H, respectively. To mimic the dynamics of the informed M--H sampler, we devise an MT-IT scheme using a mixture proposal distribution: given the current state $(\bM, p_{\mathrm{match}}, \lambda)$, with probability $\delta$ we propose a new value for $(p_{\mathrm{match}}, \lambda)$ from its conditional distribution given $\bM$, and with probability $1 - \delta$ we propose a new value for $\bM$ by uniformly sampling from $\cN_{\bM}$. We set $\delta = 0.99$ and the number of tries $m$ in the MT-IT algorithm to 100. Let $(\bM, p_{\match}, \lambda)$ denote the current state, and use $(\bM', p'_{\match}, \lambda')$ to denote the next. More details of the samplers are given as follows.

\medskip
\noindent \textbf{Uninformed M--H.} This refers to the standard Metropolis-within-Gibbs sampler. Each iteration has two steps. First, we draw $(p_{\match}', \lambda')$ from the conditional posterior distribution (which is known) given $\bM$. Second, we propose $\tilde{\bM}$ by uniform sampling from the neighborhood $\cN_{\bM}$ and accept it with probability $\min\{1, \pi(\tilde{\bM} \mid p_{\match}', \lambda') / \pi(\bM \mid p_{\match}', \lambda') \}$. Set $\bM' = \tilde{\bM}$ if the proposal is accepted, and $\bM' = \bM$ otherwise. Each iteration involves two posterior calls since we need to evaluate both $\pi(\bM \mid p_{\match}', \lambda')$ and $\pi(\tilde{\bM} \mid p_{\match}', \lambda')$.

\medskip
\noindent \textbf{Informed M--H.} This refers to the Metropolis-within-Gibbs sampler where $\bM$ is updated by a locally balanced informed proposal. In each iteration, we first draw $(p_{\match}', \lambda')$ from the conditional distribution given $\bM$. Next, we randomly select a subset of indices $\cI \subset \{1, 2, \dots, n_1\}$ and $\cJ \subset \{1, 2, \dots, n_2\}$ and apply an informed proposal with balancing function $h(r) = \sqrt{r}$ to the neighborhood subset $\{ \bM^{ij} \colon i \in \cI, \, j \in \cJ \} \subset \cN_{\bM}$, where $\bM^{ij}$ is as defined above. This blocking strategy was used in \citet{zanella2020informed} to reduce the computational cost of  informed MH updates. Let $\tilde{\bM}$ be the proposed value. To calculate the acceptance probability, we need to evaluate both $\pi(\bM^{ij} \mid p_{\match}', \lambda')$ and $\pi(\tilde{\bM}^{ij} \mid p_{\match}', \lambda')$ for each $i \in \cI$ and $j \in \cJ$. In our implementation, we set $|\cI| = |\cJ| = 10$, so each iteration involves $2 |\cI| |\cJ| + 1 = 201$ posterior calls.

\medskip
\noindent  \textbf{MT-IT.} Letting $x = (\bM, p_{\match}, \lambda)$, we propose $m$ candidate moves, denoted by $(x_1, \dots, x_m)$, such that each $x_i$ is proposed independently as follows. With probability $\delta$, we set $x_i = (\bM, \tilde{p}_{\match}, \tilde{\lambda})$ where $(\tilde{p}_{\match}, \tilde{\lambda})$ is drawn from the conditional distribution given $\bM$. With probability $1 - \delta$, we set $x_i = (\tilde{\bM}, p_{\match}, \lambda)$ where $\tilde{\bM}$ is sampled from $\cN_{\bM}$ uniformly at random. We set $\delta = 0.99$ and $m=100$ so that the MT-IT has a similar dynamic to the informed M--H. Each iteration of MT-IT involves $m = 100$ posterior calls.

\medskip 

The \texttt{italy} package contains the record data of 20 regions. To illustrate the performance of samplers, we  focus on four regions: region 1, region 4, region 18, and region 20. Region 4 and region 18 have relatively small sample sizes, with $n_1 = 310$ and $n_2 = 277$ for region 4 and $n_1 = 294$ and $n_2 = 292$ for region 18. Meanwhile, region 20 has a moderate sample size of $n_1 = 637$ and $n_2 = 641$, and region 1 has a large sample size of $n_1 = 1285$ and $n_2 = 1121$. 
The results are given in Figure~\ref{fig:linkage}. In scenarios with small sample sizes (specifically in regions 4 and 18), all three sampling methods seem to have converged after approximately $8 \times 10^5$ posterior evaluations. However, MT-IT can move to high-posterior regions more quickly than the other two. It's worth noting that although the performance of uninformed and informed M--H is very similar, the number of iterations for informed M--H is much smaller, given the fixed total number of posterior calls. When parallel computing is available, both MT-IT and informed M--H can be much more efficient. as the sample size increases (notably in regions 20 and 1), the three samplers are not able to converge after $8 \times 10^5$ posterior calls. Nevertheless, we still observe that the log-posterior curve of MT-IT increases much faster than that of the other two samplers. Interestingly, it can be seen that the number of matched pairs does not necessarily reflect the efficiency of the sampler, especially when the sampler has not converged (this statistic was also used in \citet{zanella2020informed}). Compared to M--H algorithms, MT-IT is more likely to visit some $\bM$ with fewer matched pairs but higher posterior probability in the early stages of sampling.

\begin{figure}
    \centering
    \includegraphics[width=0.48\linewidth]{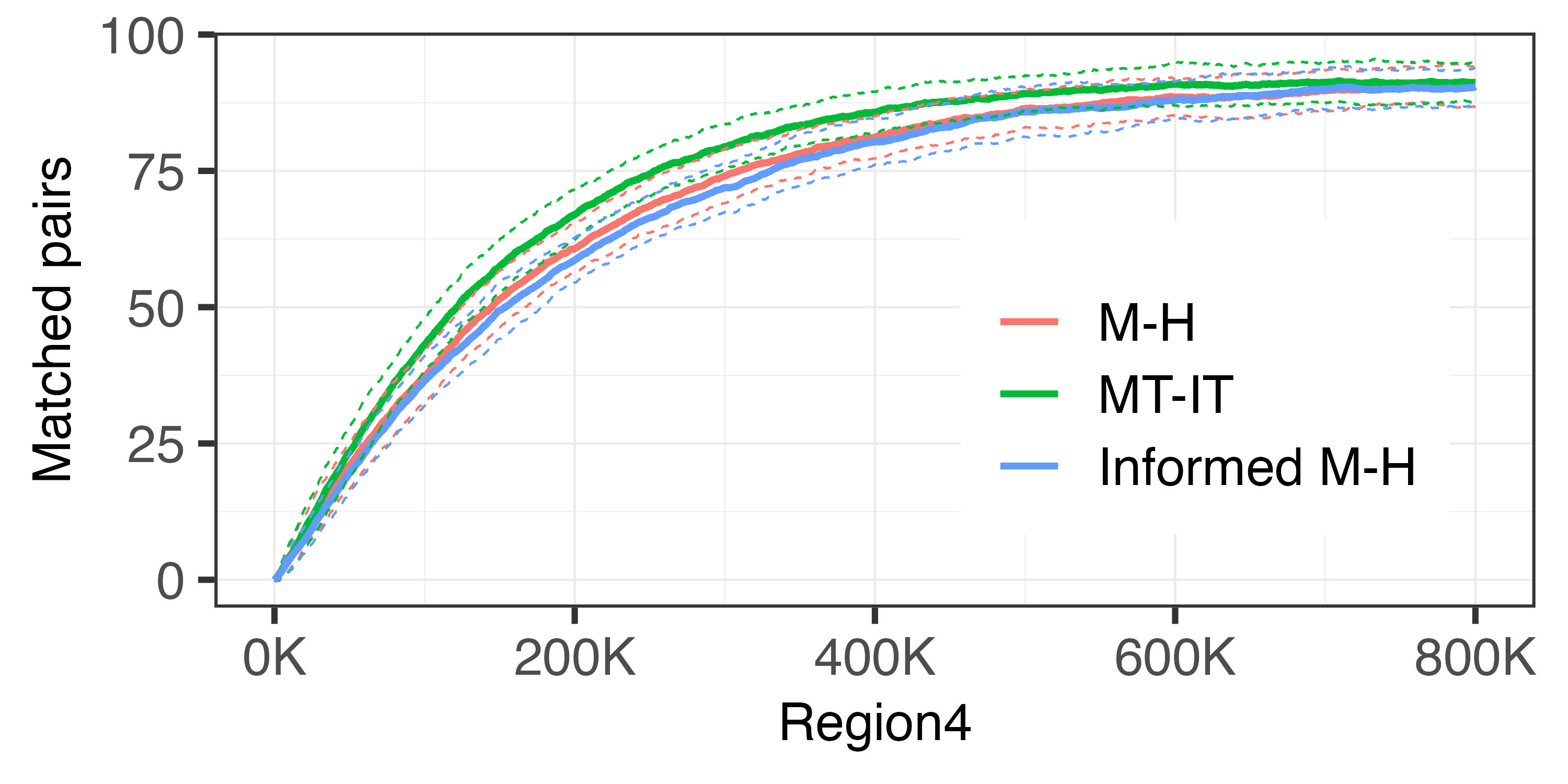}
    \includegraphics[width=0.48\linewidth]{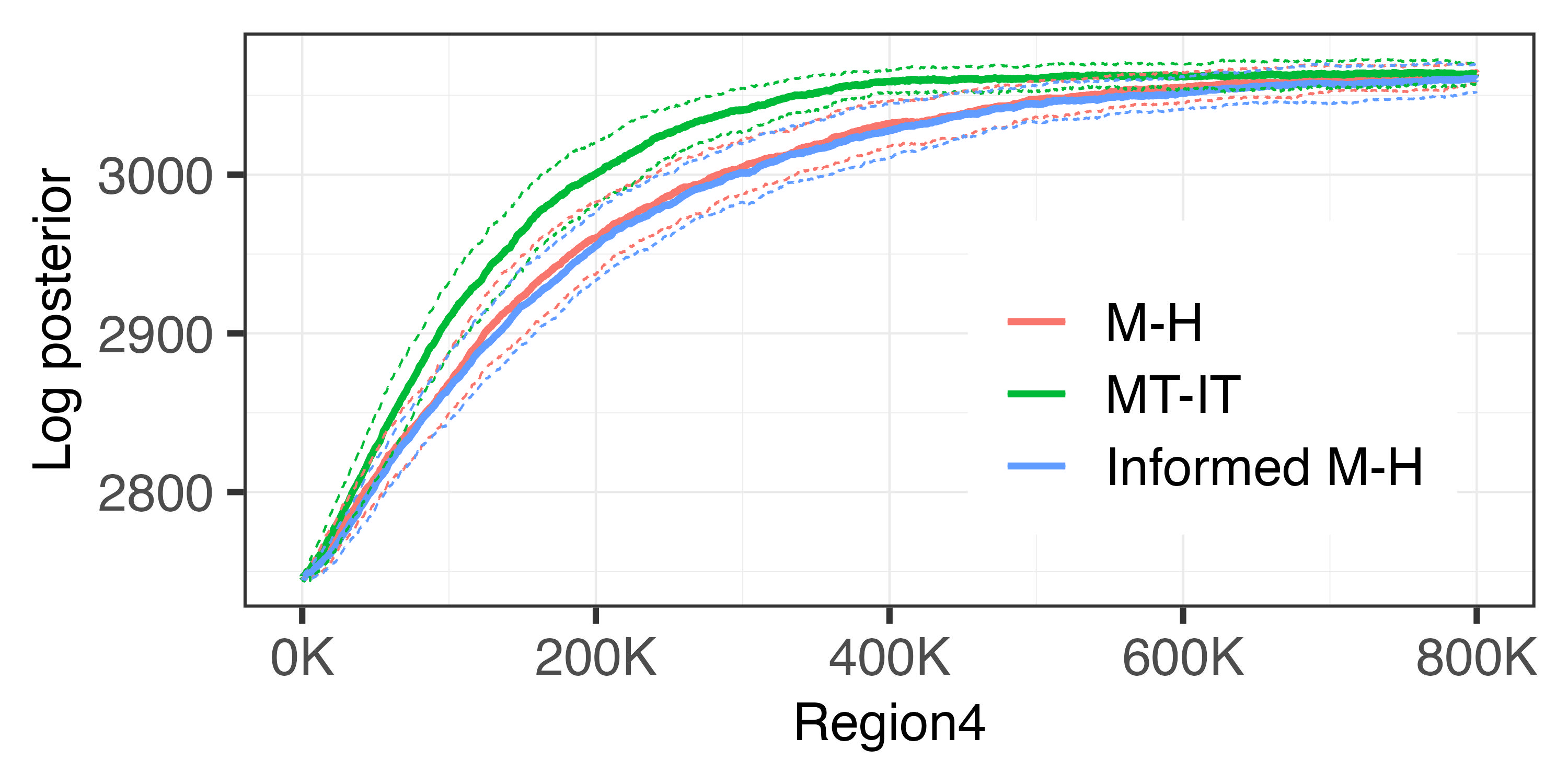}\\
    \includegraphics[width=0.48\linewidth]{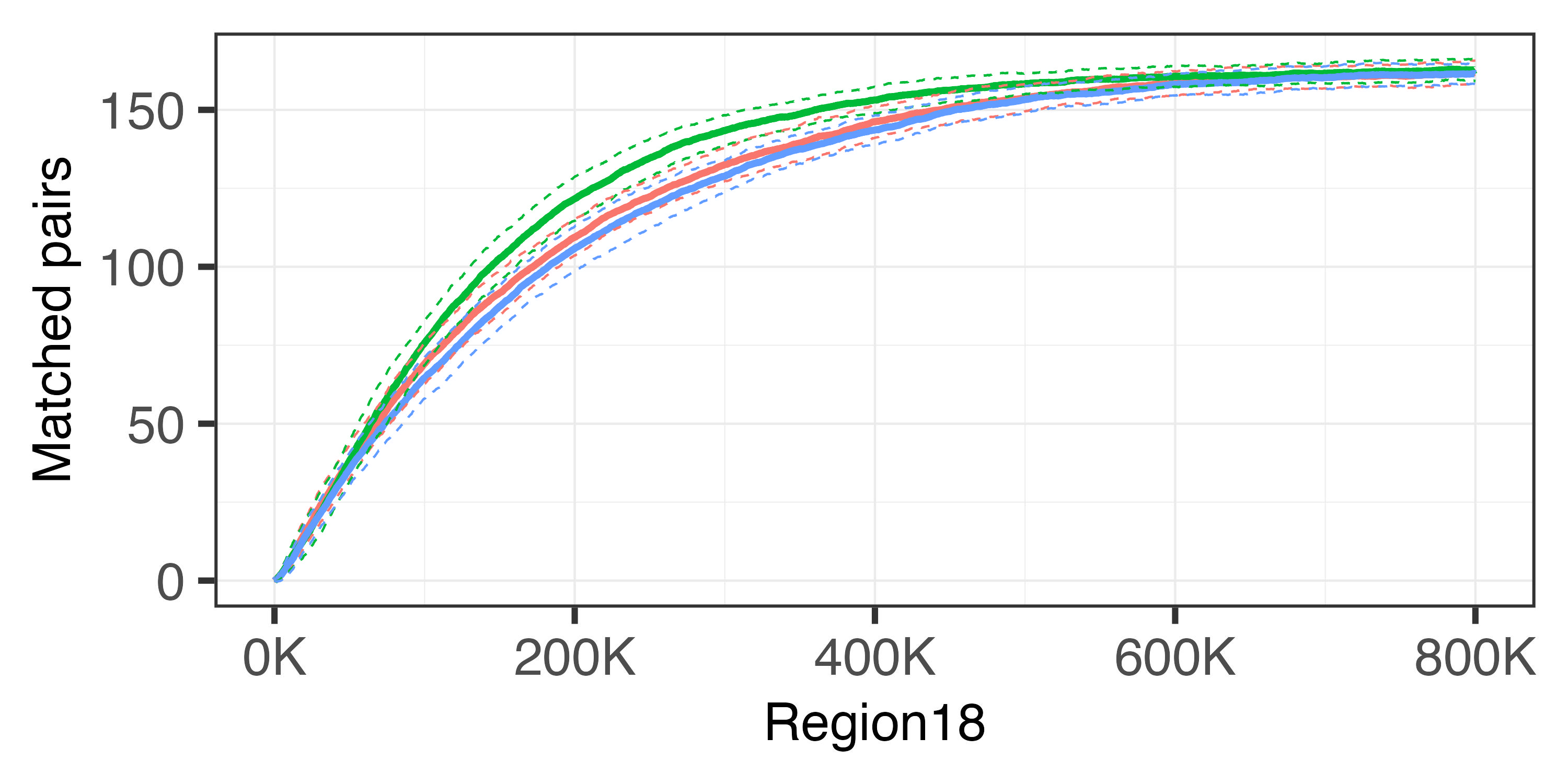}
    \includegraphics[width=0.48\linewidth]{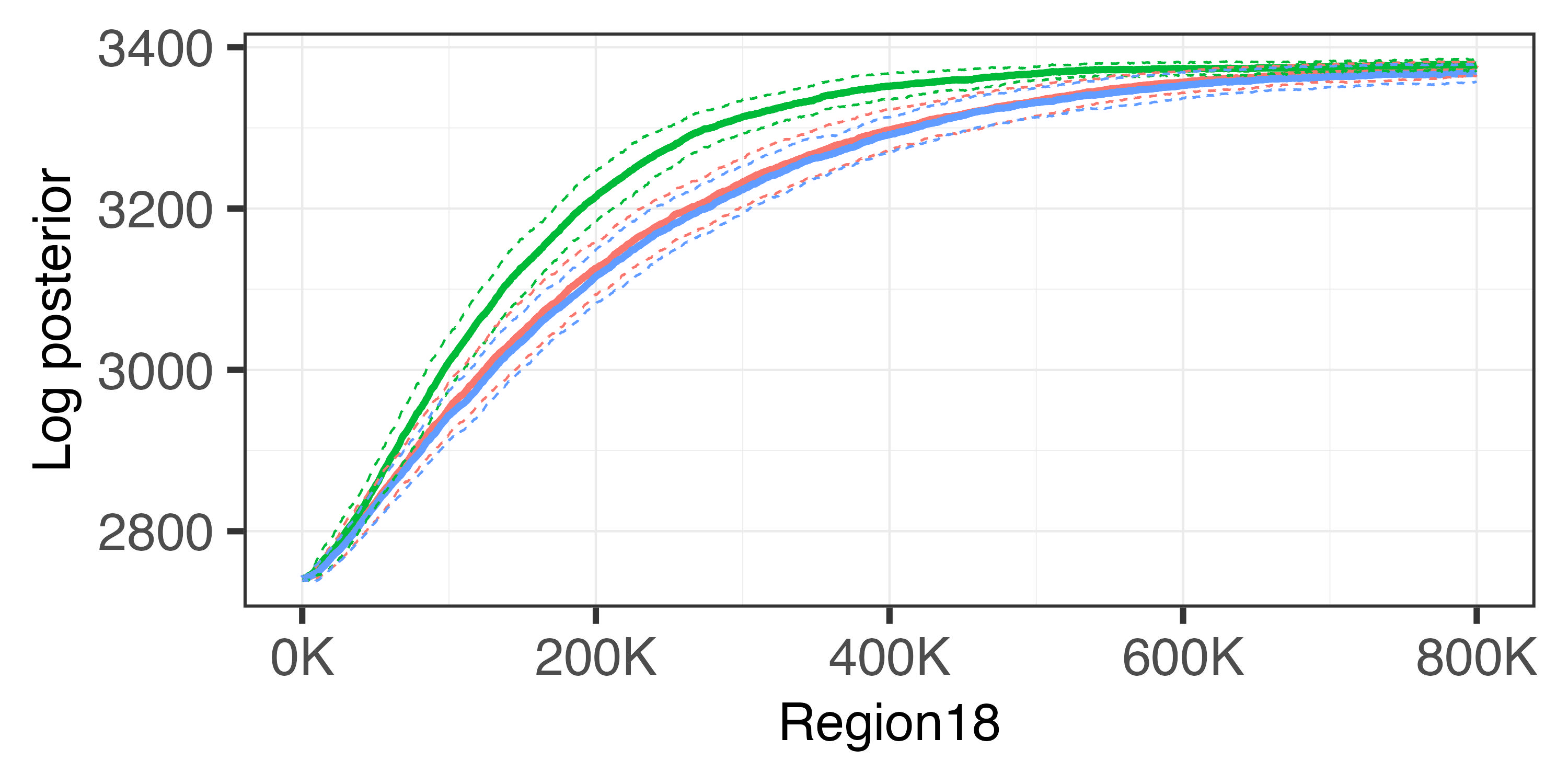} \\
    \includegraphics[width=0.48\linewidth]{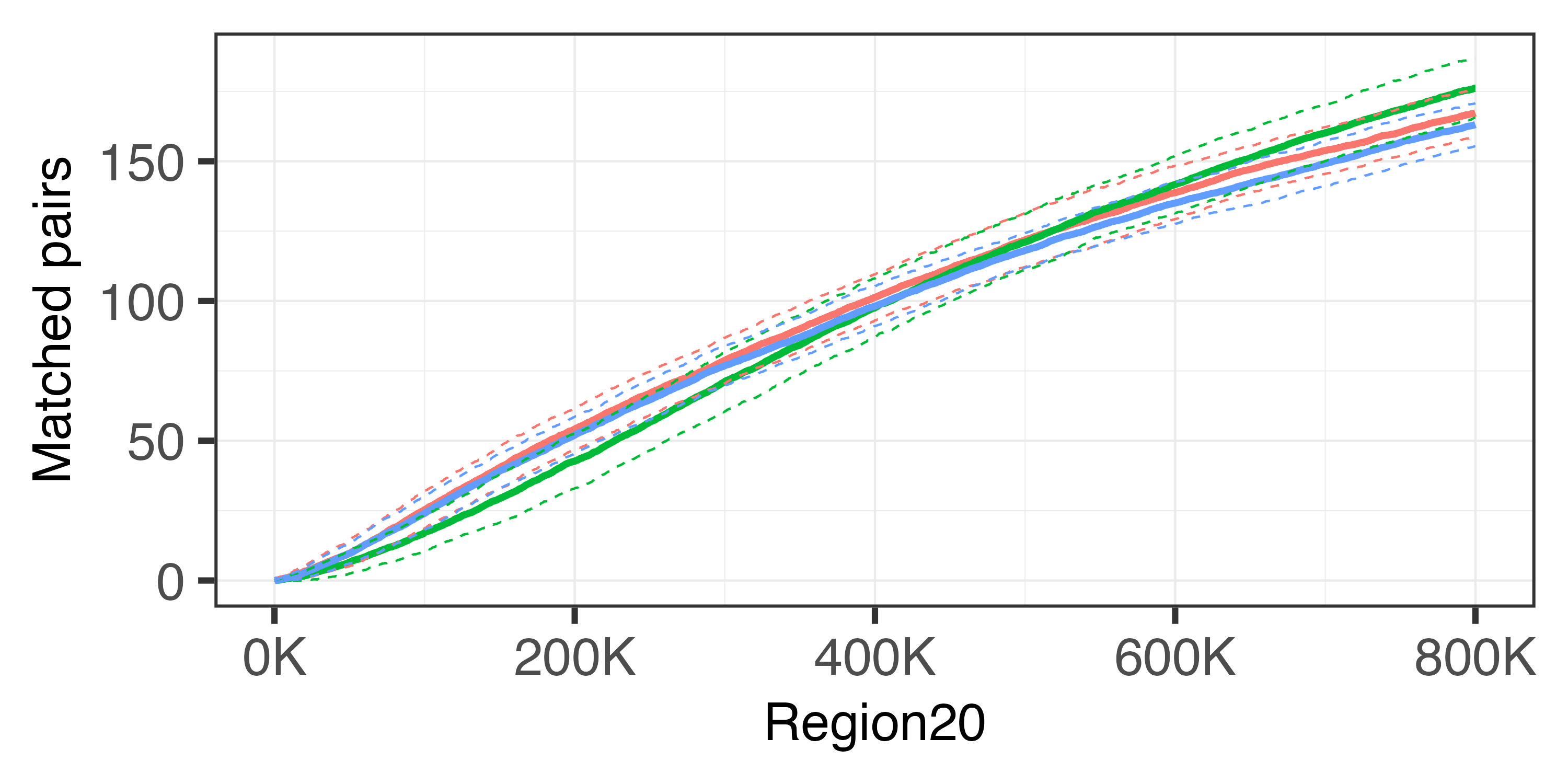}
    \includegraphics[width=0.48\linewidth]{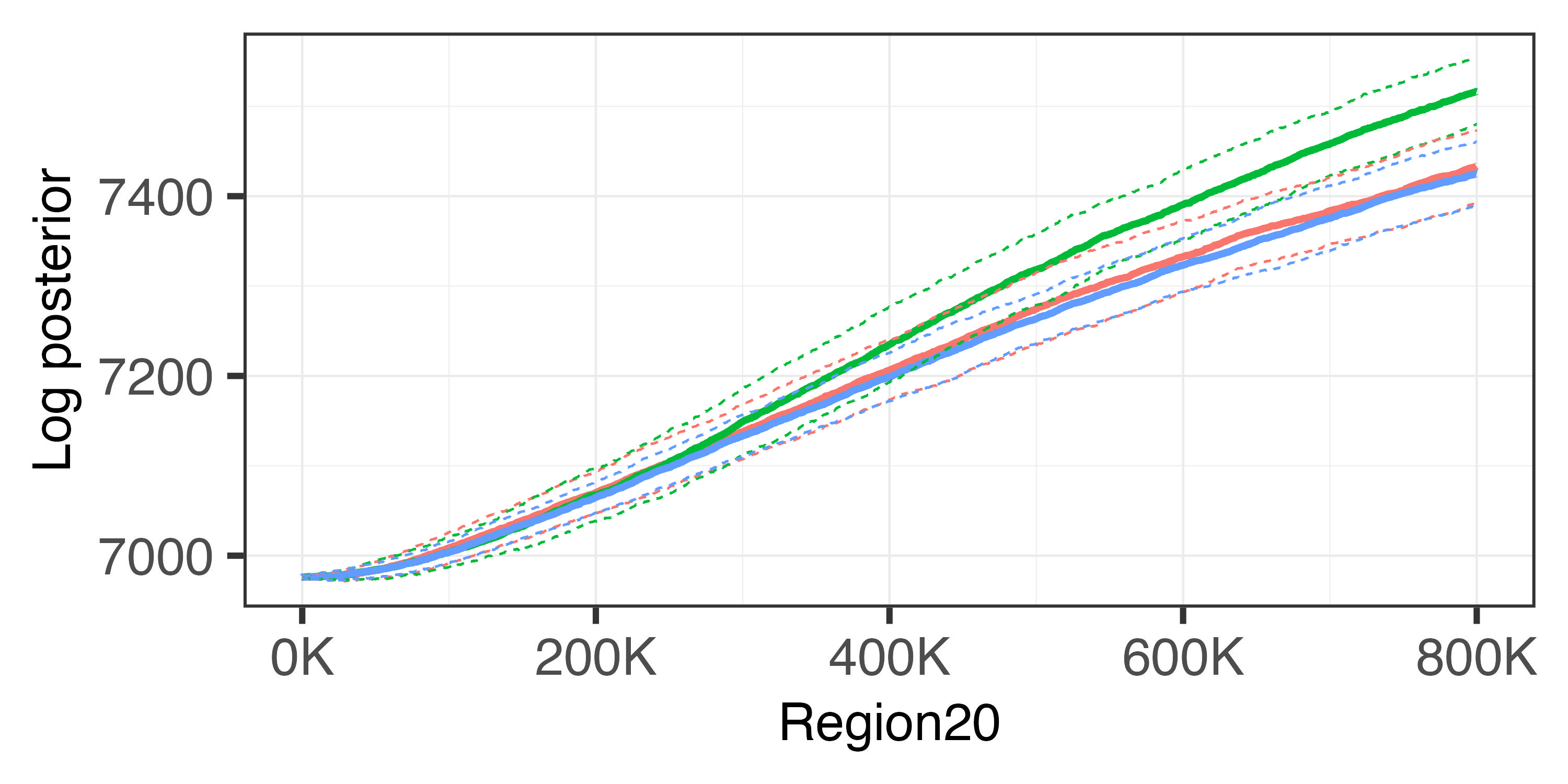} \\
    \includegraphics[width=0.48\linewidth]{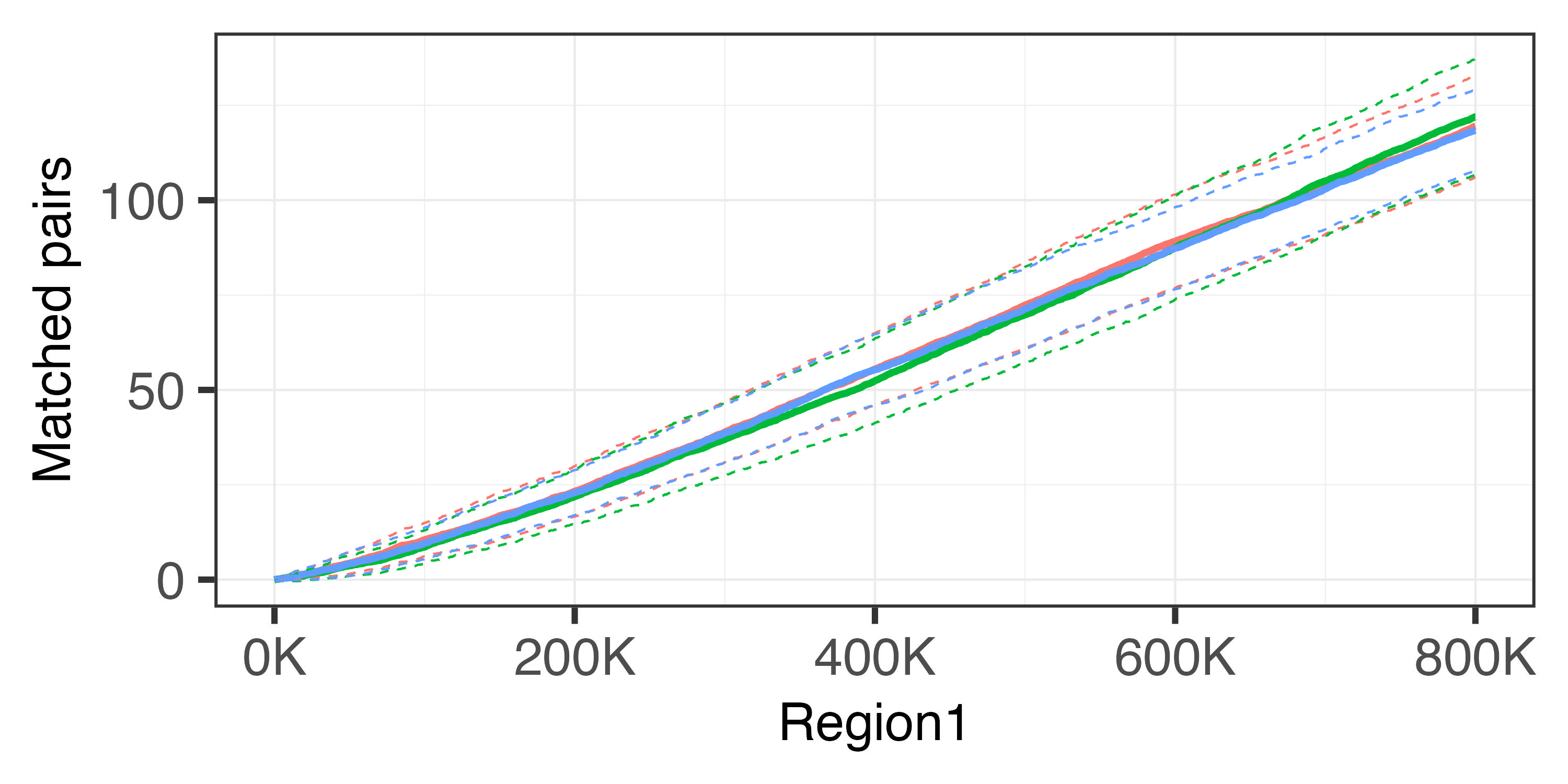}
    \includegraphics[width=0.48\linewidth]{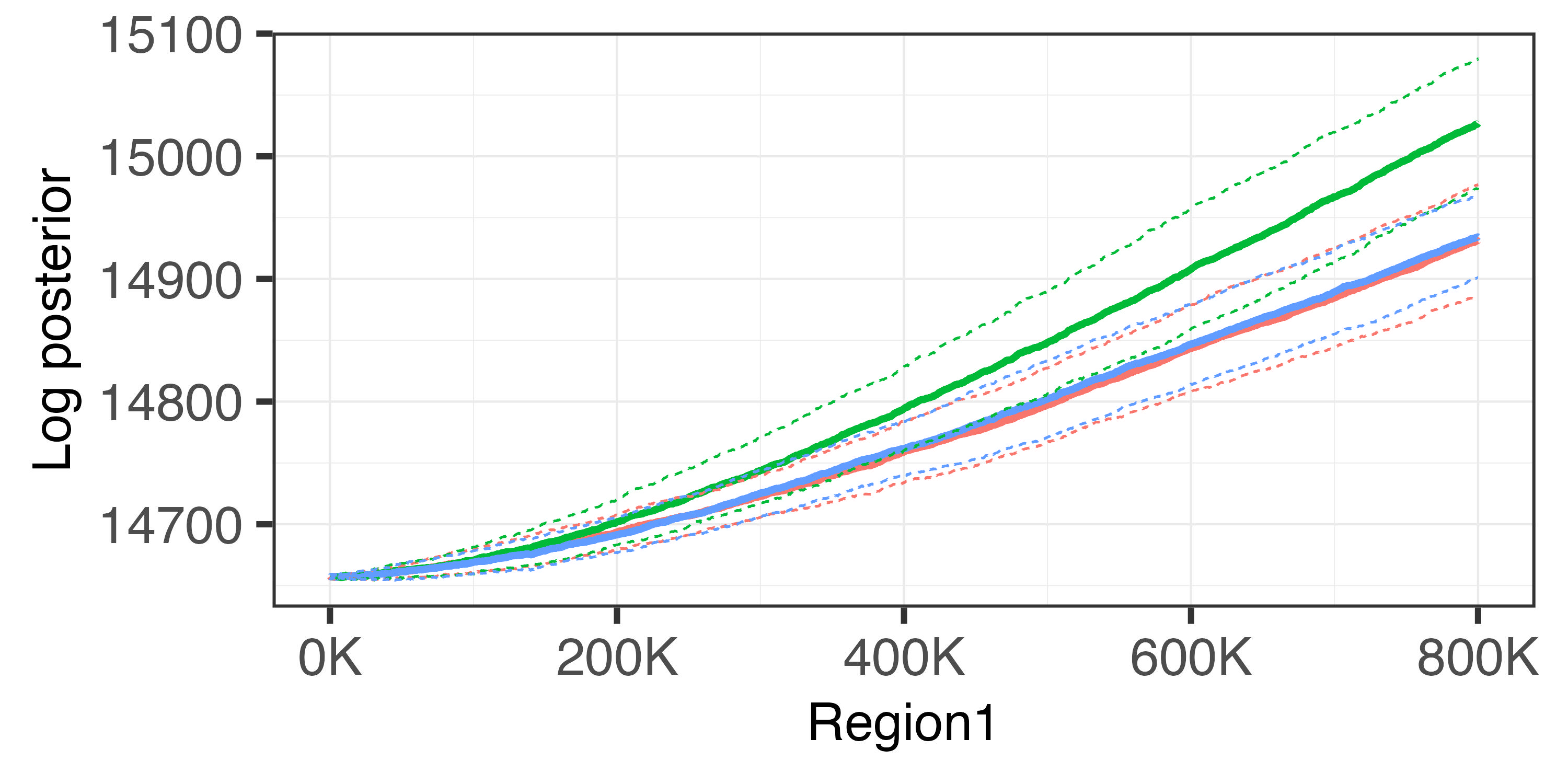} \\
    \caption{Results for the record linkage analysis of regions 1, 4, 18 and 20 in the \texttt{R} package \texttt{italy}. 
    The solid colored lines represent the average values obtained from 50 runs, and dashed lines indicate $\pm 1$ standard deviation. M--H: Metropolis-within-Gibbs sampler with uninformed MH updates; Informed M--H: Metropolis-within-Gibbs sampler with informed MH updates;  
    MT-IT: Algorithm~\ref{alg:mtm-iit}.}
    \label{fig:linkage} 
\end{figure}

\clearpage 
\newpage 

\section{Simulation Study for Section~\ref{sec:abc}}\label{app:pse}
As described in Section~\ref{sec:abc}, we study a toy example presented in Section 4.2 of~\citet{lee2014variance}, where $\pi$ is a geometric distribution with success probability $1 - ab$ for some $a, b\in (0, 1)$. For each $x \in \{1, 2, \dots, \}$,  assume that we have access to some estimator $\hat{\pi}(x)$, which can be written as $\hat{\pi}(x) = \tilde{U}_x (1-a)a^{x - 1} $, where $\tilde{U}_x$ is the mean of $K$ Bernoulli random variables with success probability $b^x$; see~\citet{lee2010utility} for how this example arises from approximate Bayesian computation. 

We use simulation to study the performance of  the pseudo-marginal  M--H algorithm and Algorithm~\ref{alg:pm-iit} with balancing function $h(r) = \sqrt{r}$. 
The neighborhood of each $x$ is simply defined to be $\cN_x = \{x-1, x+1\}$ if $x \geq 2$ and $\cN_x = \{2\}$ if $x = 1$. 
We set $a = 0.5$, $b = 0.4$ and $K = 100$, initialize both samplers at $x^{(0)} = 15$, and consider the estimation of the posterior mean $\sum_{x=1}^\infty x \,  \pi(x)$.  

The left panel of Figure~\ref{fig:pseudo.geom} shows the results when we discard first half of the samples as burn-in.  
We see that Algorithm~\ref{alg:pm-iit}  quickly produces a highly precise estimate, while the M--H scheme requires a huge number of iterations to become roughly unbiased but the variance stays large even after $10^5$ iterations.
If no samples are discarded (the right panel), Algorithm~\ref{alg:pm-iit} is still much more efficient, but both samplers are slightly biased. 

\begin{figure}[!h]
    \centering
    \includegraphics[width=0.48\linewidth]{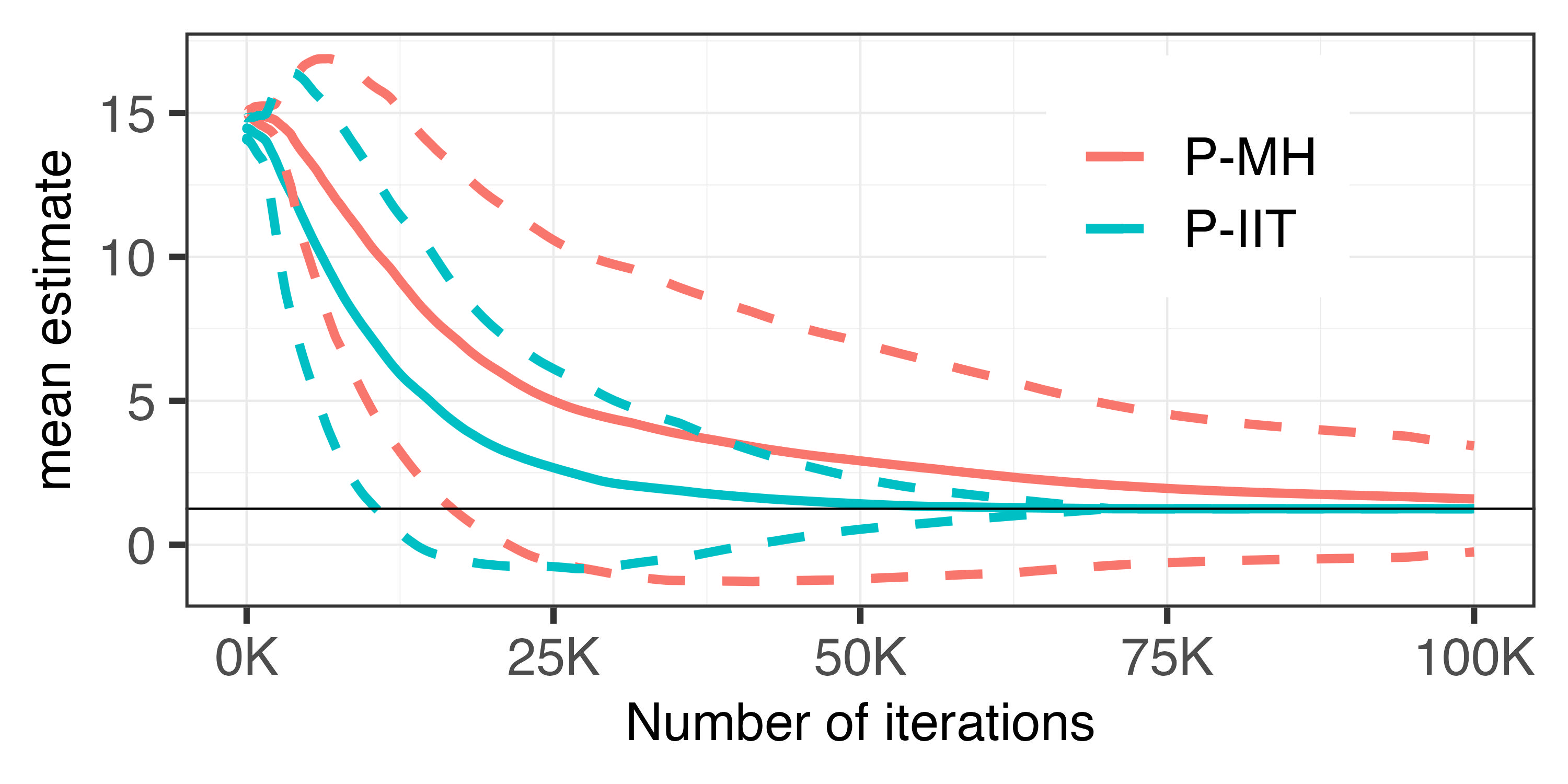}
    \includegraphics[width=0.48\linewidth]{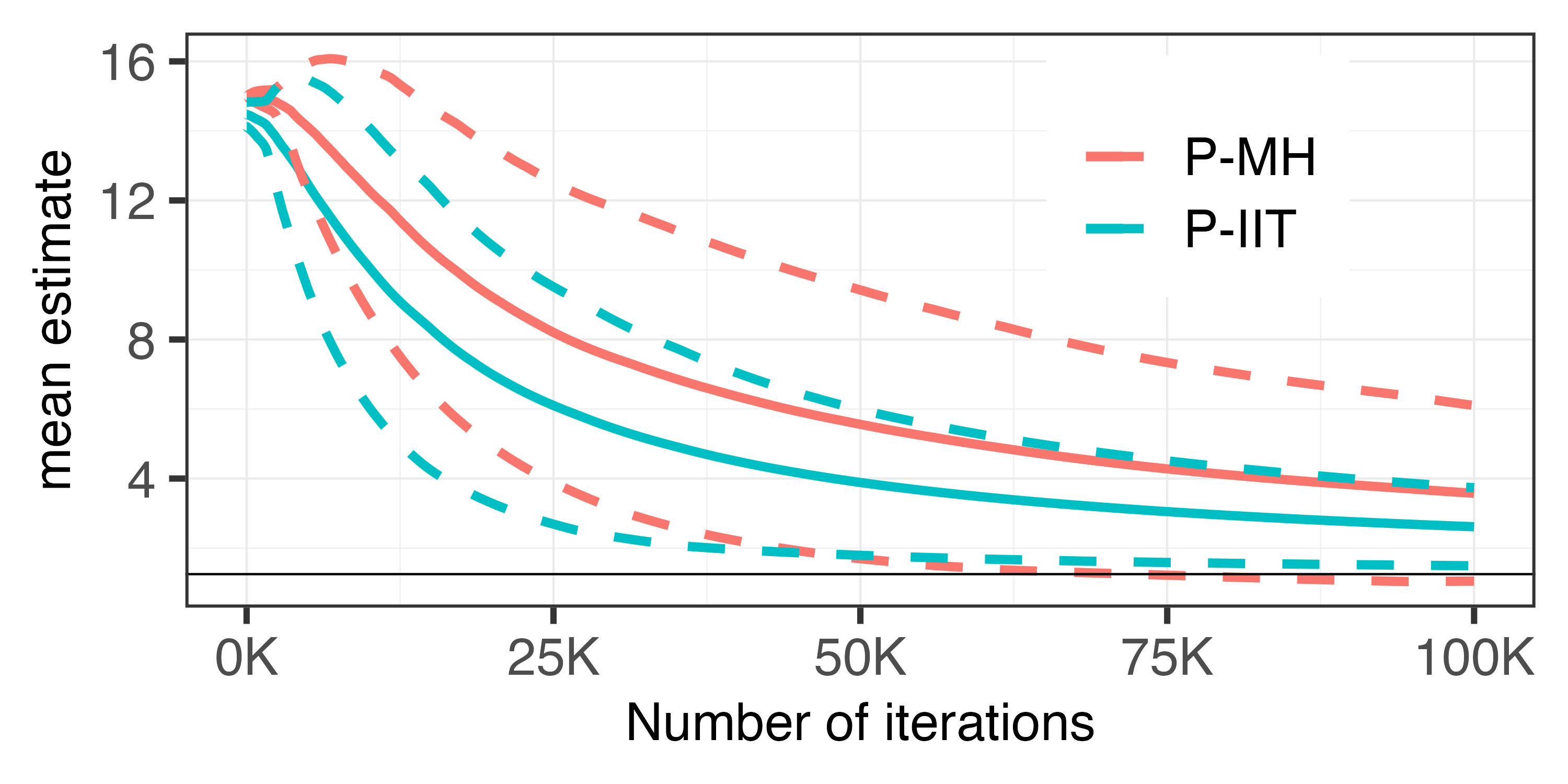}
    \caption{Estimate of the mean of $\pi$ for the geometric distribution example in Section~\ref{sec:abc}. The solid colored lines represent the MCMC estimates, and dashed lines indicate $\pm 1$ standard deviation calculated from 50 runs. The black line indicates the true value, $(1 - ab)^{-1} = 1.25$. 
    Left: half of samples discarded as burn-in when calculating the estimates; right: estimates calculated using all samples collected so far. 
    }   \label{fig:pseudo.geom} 
\end{figure}  

\smallskip
It should be noted that because the auxiliary variable $\tilde{U}_x$ can be zero, Algorithm~\ref{alg:pm-iit} is theoretically biased. We choose not to correct for this bias, since the bias is  negligible (clearly in the left panel of Figure~\ref{fig:pseudo.geom}) and we find that it can significantly help stabilize the algorithm's performance. A detailed explanation is given below.

\smallskip
\noindent \textbf{Why Algorithm~\ref{alg:pm-iit} is biased.} Denote the posterior estimate by $\hat{\pi}(x) = u \pi(x)$. As in Section~\ref{sec:pseudo.iit}, we consider the joint space $(x, u, u_-, u_+)$ such that  $\hat{\pi}(x)   = u \pi(x)$, $\hat{\pi}(x-1)= u_- \pi(x-1)$, $\hat{\pi}(x+1) = u_+ \pi(x+1)$. 
(If $x=1$, we only need to consider the state variable $(x, u, u_+)$.)
The  argument used in the proof of Lemma~\ref{lm:pm} shows that the stationary distribution of Algorithm~\ref{alg:pm-iit} is 
\begin{align*}
    \pi_{\pseu}(x, u, u_-, u_+) \propto  u \, \pi(x)  Z(x, u, u_-, u_+) g( u | x) g(u_- | x - 1) g(u_+ | x + 1), 
\end{align*}  
where $g(\cdot | x)$ denotes the distribution $(K b^x)^{-1} \mathrm{Binom}(K, b^x)$. 

The key difference from the proof of Lemma~\ref{lm:pm} is that now we have $Z(x, u, u_-, u_+) = 0$ if $u > 0$ but $u_- = u_+ = 0$.\footnote{
Assume that $h$ satisfies $h(0) = 0$, which is true for all balancing functions considered in this work. 
} (If $x=1$, then $Z(x, u, u_+) = 0$ whenever $u_+ = 0$.)
Consequently, if we still define the target distribution $\pi(x, u, u_-, u_+)$ as in Lemma 2, that is, 
\begin{align*}
    \pi(x, u, u_-, u_+) \propto  u \, \pi(x)   g( u | x) g(u_- | x - 1) g(u_+ | x + 1), 
\end{align*}
then $\pi_{\pseu}$ will have a strictly smaller support than $\pi$ and thus the importance weights of Algorithm~\ref{alg:pm-iit} are biased. 

\smallskip
\noindent  \textbf{Unbiased importance weights.} We can define our target distribution by 
\begin{align*}
\pi(x, u, u_-, u_+) = \left\{ \begin{array}{cc}
  u \, \pi(x) g(u | x) g(u_- | x - 1) g(u_+ | x + 1) B(x)^{-1},  &  \text{ if } u_- \neq 0  \text{ or } u_+ \neq 0, \\
  0,   &  \text{ if } u_- = u_+ = 0, 
\end{array} 
\right.
\end{align*}
where $B(x)$ is the probability that $u_-, u_+$ are not both equal to zero, when $u_- \sim g(\cdot | x - 1)$ and $u_+ \sim g(\cdot | x + 1)$. (If $x=1$, $B(x)$ is the probability that $u_+ > 0$.) 
One can calculate that 
\begin{align*}
    B(x) = \left\{\begin{array}{cc}
        1 - (1 - b^{x+1})^K, &  \text{ if } x = 1, \\
        1 - (1 - b^{x-1})^K (1 - b^{x+1})^K,   & \text{ if } x \geq 2 
    \end{array}
    \right. 
\end{align*}
Clearly, $\pi(x, u, u_-, u_+)$ still has $\pi(x)$ as the marginal. 
The corrected importance weight is  given by $[ Z(x, u, u_-, u_+) B(x) ]^{-1}$. 

\smallskip
\noindent  \textbf{Should we correct for the bias?} The above discussion reveals that the importance weight of Algorithm~\ref{alg:pm-iit} needs to be multiplied by a correction factor $1/B(x)$. However, observe that $1/B(x) \rightarrow \infty$ as $x \rightarrow \infty$, which can make the importance sampling estimates of P-IIT unstable. 
Indeed, we find in our simulation study that when the factor $1/B(x)$ is taken into account, the performance of P-IIT becomes quite similar to that of P-MH. 
If the factor $1/B(x)$ is not included, which means that we still use slightly biased importance weight estimates of Algorithm~\ref{alg:pm-iit}, we find that the performance of P-IIT is significantly improved, as shown in Figure~\ref{fig:pseudo.geom} in the main text. 
Intuitively, if the P-MH sampler is at some large $x$ and $\tilde{U}_x$ happens to be large as well, the chain can get stuck at $x$ for a huge number of iterations since it is very difficult to get non-zero posterior estimates for the neighboring states $x - 1$ and $x + 1$ (recall that the success probability for the binomial distribution of $\tilde{U}_x$ decreases to $0$ exponentially fast in $x$). 
This causes P-MH to have unbounded asymptotic variance~\citep{lee2014variance}. 
For P-IIT, we can force the chain to keep moving but need to estimate the holding time  which is also unbounded.  
When we do not include $1/B(x)$ in the importance weight calculation, we are essentially using a biased but much more stable estimate for the holding time, and our simulation study confirms that this leads to much better MCMC estimates in practice.

\clearpage 
\newpage 

\section{Simulation Studies for Variable Selection Targets} \label{supp:var-sel}
\subsection{Three Toy Examples} \label{app:toy-varsel}
We consider three toy examples where $\cX = \{0, 1\}^p$ and $\pi$ can be expressed in closed form. 
We define the neighborhood of $x$ as $\cN_x = \{y \colon \norm{x - y}_1 = 1\}$, where $\norm{\cdot}_1$ denotes the $\ell^1$-norm. 
We say $x \in \cX$ is a local mode if $\pi(x) > \pi(y)$ for every $y \in \cN_x$, and $\pi$ is unimodal if there is only one local mode. Despite being idealized and contrived, the three examples illustrate typical scenarios in variable selection: $\pi$ is unimodal with independent coordinates in Example~\ref{ex:vs1}, unimodal with dependent coordinates in Example~\ref{ex:vs2}, and bimodal in Example~\ref{ex:vs3}. 

For each example, we introduce a discrete-valued function $F(x)$ and measure the convergence rates of MCMC samplers using $d(\pi, \hat{\pi}_T; F) \in [0, 2]$, which is defined by
\begin{equation}\label{eq:toy.metric}
    d(\pi, \hat{\pi}_T; F) = \sum_{k} \left| \pi( F^{-1}(k) ) - \hat{\pi}_T (F^{-1}(k)) \right|, \quad F^{-1}(k) = \{x \colon F(x) = k\},
\end{equation}
where the summation is over all possible values of $F$ and $\hat{\pi}_T$ is the (importance-weighted) empirical distribution of $T$ MCMC samples. Thus, $d(\pi, \hat{\pi}_T; F)$ is the total variation distance between the push-forward measures $\pi \circ F^{-1}$ and $\hat{\pi}_T \circ F^{-1}$. The closed-form expression of $\pi$ enables us to exactly calculate $d(\pi, \hat{\pi}_T; F)$. We do not consider the total variation distance between $\pi$ and $\hat{\pi}_T$ since summation over $\cX$ is not computationally feasible. 

\begin{example}\label{ex:vs1}
Let $p_1 \in \{0, 1,  \dots, p\}$ and $x^* \in \bbR^p$ be given by $x^*_i = 1$ if $i \leq p_1$ and $x^*_i = 0$ if $i > p_1$. For $\theta > 0$, define the target distribution by
$$\pi(x) = \frac{\exp(-\theta \norm{x - x^*}_1)}{\normc_\theta}, \quad \text{where} \ \normc_\theta = (1 + e^{-\theta})^p.$$ 
Clearly, $\pi$ is unimodal with mode $x^*$, and $\pi$ has independent coordinates. This example represents an ideal variable selection problem where $x^*$ is the true model, all covariates in $x^*$ have equally strong effects, and the design matrix has orthogonal columns. The parameter $\theta$ controls the tail decay rate and indicates the signal strength.
\end{example}

\begin{example}\label{ex:vs2} 
\rm For $\theta > 0$, define $\pi(x) = \exp(- \theta \ell(x)) / \normc_\theta$ where
\begin{equation}\label{eq:pi.vs2}
\begin{aligned}
     \ell(x) = \left\{
     \begin{array}{cc}
     \norm{x}_1 - 1, & \text{if } x_1 = 1, \\ 
     2p - \norm{x}_1, & \text{if } x_1 = 0,
     \end{array}
     \right. \quad \text{and} \quad \normc_\theta = \left(1 + e^{-\theta(p + 1)}\right)(1 + e^{-\theta})^{p - 1}.
\end{aligned}
\end{equation}
It is not difficult to see that $\pi$ is unimodal with mode $x^* = (1, 0, 0, \dots, 0)$, but $\pi$ has dependent coordinates. In variable selection, such a posterior distribution arises when only the first covariate has a nonzero effect but is correlated with all the other covariates.
\end{example}

\begin{example}\label{ex:vs3} 
\rm Let $p_1 \in \{1, 2, \dots, p\}$. Let $x_{(1)}^*$ and $x_{(2)}^* \in \cX$ be given by $x_{(1)}^* = (1, 0, 1, \dots, 1, 0, \dots, 0)$ and $x_{(2)}^* = (0, 1, 1, \dots, 1, 0, \dots, 0)$ such that $\norm{x_{(1)}^*}_1 = \norm{x_{(2)}^*}_1 = p_1$. For $\theta > 0$, define
\begin{equation}\label{eq:pi.toy3}
    \pi(x) = \frac{\exp(-\theta \norm{x - x_{(1)}^*}_1) + \exp(-\theta \norm{x - x_{(2)}^*}_1)}{\normc_\theta}, \quad \text{where} \ \normc_\theta = 2 (1 + e^{-\theta})^p.
\end{equation}
Clearly, $\pi$ is bimodal with local modes $x_{(1)}^*$ and $x_{(2)}^*$. This corresponds to a variable selection problem where the first two covariates are highly correlated.
\end{example}

\subsubsection{Simulation Results for Example~\ref{ex:vs1}}  
We fix $p = 500$, set $F(x) = \norm{x - x^*}_1$, and use simulation to find the number of evaluations of $\pi$ needed to achieve $d(\pi, \hat{\pi}_T; F) \leq 0.1$. We consider six MCMC samplers: uninformed M--H, Algorithms~\ref{alg:iit}, \ref{alg:mh-iit}, and \ref{alg:rn-iit}, locally balanced MTM of \citet{changrapidly, gagnon2022improving}, and weighted TGS of \citet{zanella2019scalable}. All samplers are initialized at $x^{(0)} = (0, \dots, 0)$. Details of their implementation are given below. 

\begin{figure}
    \centering
    \includegraphics[width=0.75\linewidth]{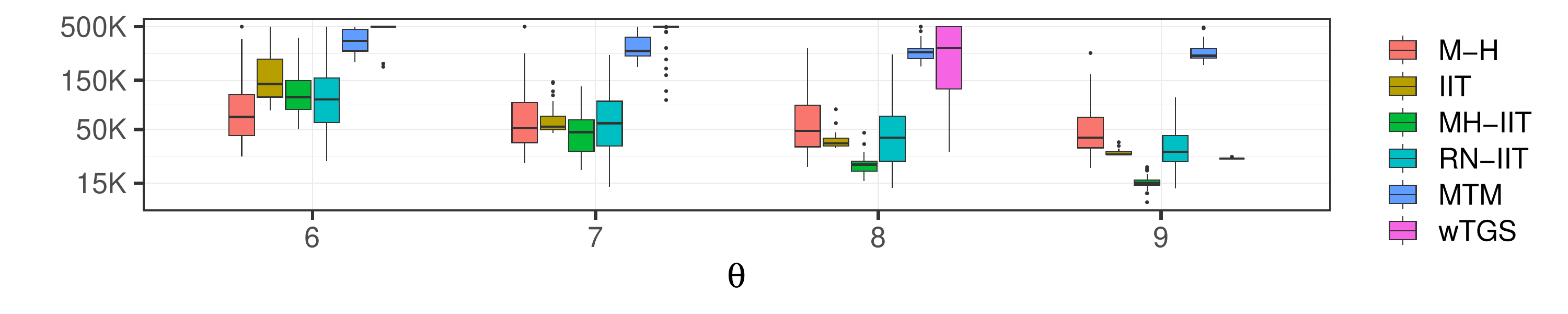}\\
    \includegraphics[width=0.75\linewidth]{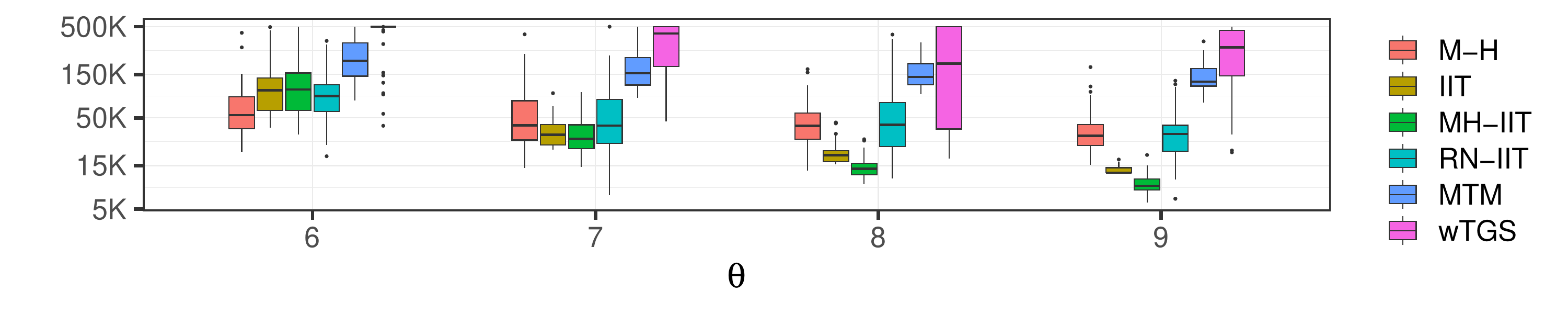}
\caption{Box plot for the number of evaluations of $\pi$ needed to accurately approximate $\pi$ in Example~\ref{ex:vs1} with $p=500$.  Top panel: $p_1=50$; bottom panel: $p_1=20$.   For each $\theta$, we repeat the simulation $50$ times. 
M--H: uninformed M--H; IIT: Algorithm~\ref{alg:iit}; MH-IIT: Algorithm~\ref{alg:mh-iit}; RN-IIT Algorithm~\ref{alg:rn-iit};
MTM: the method of~\citet{changrapidly}; wTGS: the method of~\citet{zanella2019scalable}. 
All samplers are forced to stop after $500$ thousand evaluations of $\pi$. 
} 
\label{fig:vs1}
\end{figure}
 
M--H denotes the uninformed M--H algorithm that proposes the next state by uniformly sampling from $\cN_x$. IIT denotes   Algorithm~\ref{alg:iit}, with balancing function $h(r) = \sqrt{r}$. 
MH-IIT is Algorithm~\ref{alg:mh-iit} with $\rho(x) \equiv 0.025$ and $h(r) = 1 \wedge r$. 
RN-IIT is Algorithm~\ref{alg:rn-iit} with $m = 100$ and  $h(r) = \sqrt{r}$.  
MTM is the locally balanced multiple-try Metropolis algorithm of \citet{changrapidly} with balancing function $h(r) = \sqrt{r}$ and number of tries $= 100$; the number of tries plays the same role as the parameter $m$ in RN-IIT. 
Finally, wTGS denotes the weighted tempered Gibbs sampler of \citet{zanella2019scalable} with the default parameter setting in the authors' code.

Figure~\ref{fig:vs1} shows the results when $p_1 = 50$ or $p_1 = 20$.  It can be seen from Figure~\ref{fig:vs1} that when $\theta$ is small (i.e., $\pi$ is relatively flat), uninformed M--H is the best, as it is efficient at exploring the whole space. When $\theta$ is large, IIT, MH-IIT, and RN-IIT outperform uninformed M--H due to the importance weighting. 
In particular, MH-IIT has the best performance among all methods. RN-IIT seems less efficient than IIT and MH-IIT in this scenario, but it is much better than the locally balanced MTM method, an M--H algorithm with very similar dynamics. The wTGS sampler only performs well when $\theta$ is large, which can be partially explained by the theory developed in \citet{zhou2022rapid}. 

We have also tried the Hamming ball sampler of \citet{titsias2017hamming} and the LIT-MH sampler of \citet{zhou2021dimension}, but both methods performed poorly according to our metric due to the high per-iteration computational cost (these algorithms could be beneficial if parallel computing is available).

\smallskip
\noindent \textbf{On Rao-Blackwellization.} In the variable selection literature, Rao-Blackwellization is often used to improve the posterior estimates~\citep{guan2011bayesian, zhou2019fast, zhou2021dimension, griffin2021search}, and it can be implemented for IIT and wTGS at almost no extra computational cost. 
If the posterior is unimodal with independent coordinates, Rao-Blackwellization can drastically improve the sampler's performance. 
But for the purpose of fair comparison, no Rao-Blackwellization procedure is performed for any method in our simulation studies on the three toy examples. 
Another reason is that, according to our experiments, there is no significant gain if the posterior distribution has strong dependence among the coordinates.

\smallskip
\noindent  \textbf{Remarks on wTGS.} \citet{zanella2019scalable} proposed two importance tempering methods for variable selection, TGS (tempered Gibbs sampler) and its weighted version, wTGS. 
As shown in~\citet{zhou2022rapid}, TGS is essentially an IIT sampler with balancing function $h(r) = 1 + r$. It was proved in~\cite{zhou2022rapid} that this choice of balancing function is overly aggressive and, consequently, TGS can be highly efficient when the target is unimodal with very light tails, but can  underperform uninformed MH in other cases. 
wTGS uses the following proposal weighting scheme. Given $x \in \cX$ and $y \in \cN_x$, if $\norm{y}_1 = \norm{x}_1 - 1$, the un-normalized proposal weight of $y$ is set to $1$; 
if $\norm{y}_1 = \norm{x}_1 + 1$, the un-normalized proposal weight of $y$ is set to $\pi(y) / \pi(x)$. 
This modification can provide significant improvement over TGS in some scenarios, but for Example~\ref{ex:vs1}, wTGS seems to suffer from the same issue as TGS: its informed proposal scheme is too aggressive compared to, for example, a locally balanced proposal with $h(r) = \sqrt{r}$.

\subsubsection{Simulation Results for Example~\ref{ex:vs2}}  
For our simulation study, we use $p = 500$ and define $F(x) = (\norm{x}_1 - 1) \ind_{\{x_1 = 1\}} + p \ind_{\{x_1 = 0\}}$. All samplers are initialized at the state $x^{(0)}$ such that $x^{(0)}_i = 1$ if and only if $i > p - 10$. The settings of MH, RN-IIT, MTM, and wTGS are the same as in Example~\ref{ex:vs1}. 
Unlike in Example~\ref{ex:vs1}, we consider two MH-IIT schemes here: MH-IIT-1 uses $h(r) = 1 \wedge r$ (as in Example~\ref{ex:vs1}), and MH-IIT-2 uses the following balancing function with $c = 2 \theta$:
\begin{equation}\label{eq:def.hc.supp}    
h_c(r) = (1 \wedge r e^{-c} ) \vee (r \wedge e^{-c} ),  \quad c \geq 0. 
\end{equation}     
Both MH-IIT-1 and MH-IIT-2 still use  $\rho(x) \equiv 0.025$ 
MH-IIT-1 is much more conservative than MH-IIT-2: if the current state $x$ has $x_1 = 0$, $h(r) = 1 \wedge r$ assigns the same proposal weight to any $y \in \cN_x$ such that $\norm{y}_1 = \norm{x}_1 + 1$, while $h_c$ favors flipping $x_1$ since $h_c(r) = 1 \wedge r e^{-c}$ for any $r \geq 1$. See Remark~\ref{rmk:hc} below for a more detailed discussion. 
 
Figure~\ref{fig:vs2} shows the number of posterior calls needed to achieve $d(\pi, \hat{\pi}_T; F) \leq 0.2$. As in Example~\ref{ex:vs1}, uninformed MH performs well only when $\theta$ is small. For large $\theta$, $\pi$ concentrates at the mode $x^*$, so the convergence of IIT schemes largely depends on when the first covariate is selected, making MH-IIT-2 much more efficient than MH-IIT-1. It should be emphasized that balancing functions such as $h_c$ are not desirable in MH schemes, because the acceptance probability for any neighboring state with a smaller posterior (than the current one) is extremely low, causing the sampler to remain at a local mode for numerous iterations. However, this limitation is lifted in MH-IIT, and the optimal balancing function is no longer $h(r) = 1 \wedge r$. We acknowledge that the use of $c = 2\theta$ is unrealistic as it requires knowledge about $\pi$, but we have observed that using any $c > \theta$ provides gains over MH-IIT-1. Another interesting observation is that RN-IIT performs better than MH-IIT-1 in Figure~\ref{fig:vs2}, which is different from Example~\ref{ex:vs1}.

\begin{figure}
    \centering
    \includegraphics[width=0.98\linewidth]{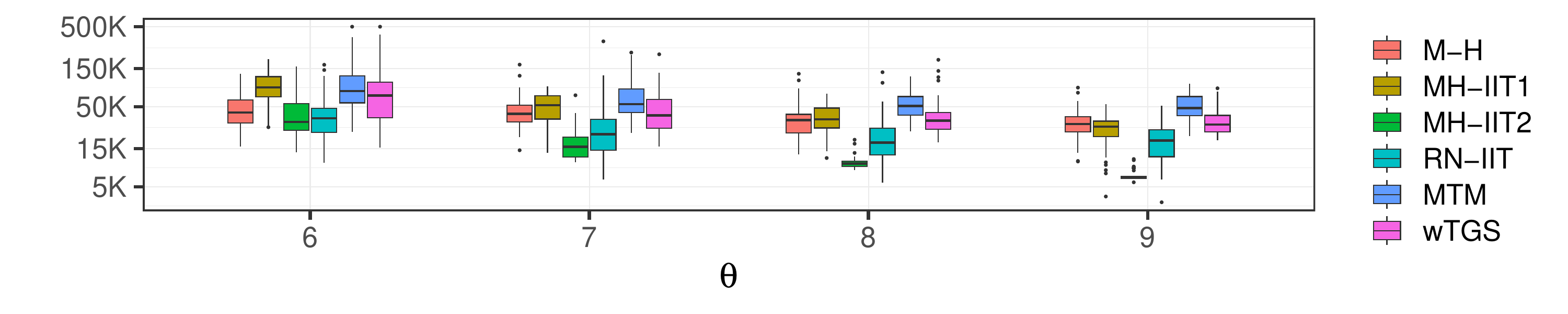}
    \caption{ 
    Box plot for the number of evaluations of $\pi$ (truncated at 500K) needed to accurately approximate $\pi$ in Example~\ref{ex:vs2} with $p=500$.  Each box represents $50$ runs.  MH-IIT-1: Algorithm~\ref{alg:mh-iit} with $h(r)  = 1 \wedge r$;  MH-IIT-2: Algorithm~\ref{alg:mh-iit} with $h_c$ in~\eqref{eq:def.hc}.  The other algorithms are the same as in Figure~\ref{fig:vs1}.  
    }\label{fig:vs2}
\end{figure}

\begin{remark}\label{rmk:hc}
\rm Equivalently, we can express $h_c$ by
\begin{align*}
h_c(r) = \left\{
\begin{array}{cc}
1, & \text{if } r \geq e^c, \\
r e^{-c}, & \text{if } 1 \leq r < e^c, \\
e^{-c}, & \text{if } e^{-c} < r < 1, \\
r, & \text{if } r \leq e^{-c}.
\end{array}
\right.
\end{align*}
It is straightforward to check that $h_c$ is indeed a balancing function for any $c \geq 0$. When $c = 0$, we get $h_c(r) = 1 \wedge r$. As $c \rightarrow \infty$, we get $e^c h_c(r) = 1 \vee r$. Hence, a larger $c$ implies more aggressive behavior of the corresponding informed proposal.

According to the definition of $\pi$ in Example~\ref{ex:vs2}, for any $x \in \cX$ and $y \in \cN_x$, we have $|\log \frac{\pi(y)}{\pi(x)}| \geq \theta$. Hence, for this particular example, using $c = 0$, $c = \theta$, or any other value in $(0, \theta)$ does not make a difference. In particular, if $x_1 = 0$, according to the definition of $\pi$, we have
\begin{align*}
\frac{\pi(y)}{\pi(x)} = \left\{
\begin{array}{cc}
\exp(\theta), & \text{if } y \in \cN_x, \norm{y}_1 = \norm{x}_1 + 1, y_1 = 0, \\
\exp(2\theta (p - \norm{x}_1)), & \text{if } y \in \cN_x, y_1 = 1.
\end{array}
\right.
\end{align*}
Hence, using $h_c$ with $c > \theta$ encourages the sampler to add the first covariate to $x$ in this case. Of course, in practice, we usually do not have much knowledge about $\pi$ for tuning the parameter $c$. The primary purpose of this simulation study is to illustrate that for MH-IIT schemes, it is often desirable to use a balancing function more aggressive than the default one, $h(r) = 1 \wedge r$.
\end{remark}

\subsubsection{Simulation Results for Example~\ref{ex:vs3}}  
We set $p = 200$ in our simulation study and let $F$ be the vector-valued function given by $F(x) = (\norm{x - x^*_{(1)}}_1, \norm{x - x^*_{(2)}}_1)$. We consider six MCMC samplers investigated in Example~\ref{ex:vs1}, initialized at $x^{(0)} = (0, \dots, 0)$. Regarding the parameter settings, the only difference is that we now use $m = 40$ for RN-IIT, since $p = 200$ in this example while $p = 500$ in Example~\ref{ex:vs1}. 

Figure~\ref{fig:vs3} shows the number of evaluations of $\pi$ needed to achieve $d(\pi, \hat{\pi}_T; F) \leq 0.5$. The performance of uninformed MH quickly deteriorates as $\theta$ increases, while the performance of IIT methods remains relatively stable. We note that, unlike in the other two examples, wTGS has very good performance in this bimodal setting. Indeed,   when $p_1 = 20$, wTGS has the best performance once $\theta$ becomes sufficiently large. 
We provide an intuitive explanation here. Observe that a sampler can move from $x_{(1)}^*$ to $x_{(2)}^*$ by either removing the first covariate from $x_{(1)}^*$ or adding the second covariate to $x_{(1)}^*$. When $p_1 \ll p$, it is much easier for the sampler to remove the right covariate than to add the right one. Compared to IIT methods, the informed proposal used by wTGS favors removing covariates from $x_{(1)}^*$, which makes wTGS able to move between the two local modes more easily than IIT methods.

\begin{figure}
    \centering
    \includegraphics[width=0.98\linewidth]{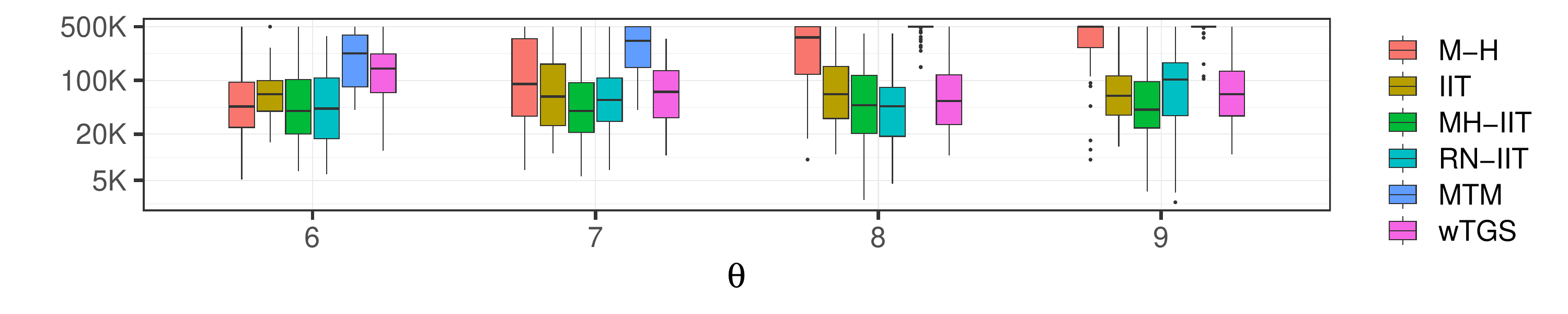} 
    \includegraphics[width=0.98\linewidth]{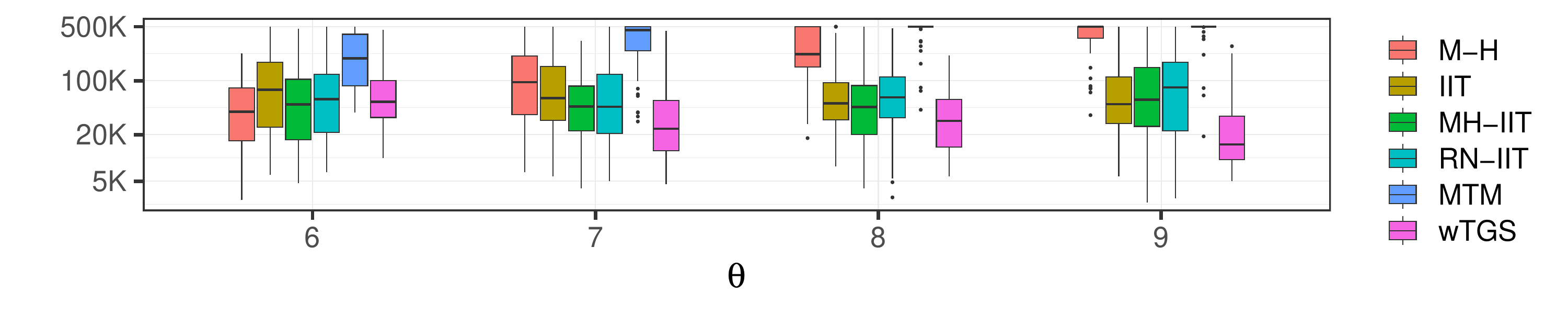}
    \caption{
    Box plot for the number of  evaluations of $\pi$ (truncated at 500K) needed to accurately approximate $\pi$ in Example~\ref{ex:vs3} with $p=200$. Top panel: $p_1 = 50$; bottom panel: $p_1 = 20$. Each box represents $50$ runs.  
    MH: uninformed MH; IIT: Algorithm~\ref{alg:iit}; MH-IIT: Algorithm~\ref{alg:mh-iit}; RN-IIT Algorithm~\ref{alg:rn-iit};
    MTM: the method of~\citet{changrapidly}; wTGS: the method of~\citet{zanella2019scalable}. 
    }    \label{fig:vs3} 
\end{figure}

\subsection{Additional Details for the Variable Selection Simulation} \label{supp:var-sel-real}
Here we provide more details for the simulation study visualized in Fig. 2 in the main text. 
This study considers the following nine MCMC samplers. M--H, IIT, RN-IIT, MTM, and wTGS are the same as in Example~\ref{ex:vs1}. In particular, the parameter $m$ in RN-IIT and the number of tries in MTM are set to $100$, and for IIT, RN-IIT, and MTM, the balancing function is $h(r) = \sqrt{r}$. For the two MH-IIT methods, we fix $\rho = 0.0025$. For MH-IIT-1, we use the conservative balancing function $h(r) = \min\{1, r\}$, and for MH-IIT-2, we use the balancing function $h$ defined in \eqref{eq:def.hc.supp} with $c = 8$. For HBS, we utilize the code provided in \citet{zanella2019scalable}, which considers a Hamming ball with radius 1 centered at the current model in each iteration. For LIT-MH, we use the parameter setting of the LIT-MH-2 algorithm presented in \citet{zhou2021dimension}. 

\newpage
\bibliographystyle{plainnat}
\bibliography{references}
\end{document}